\documentclass[conference]{IEEEtran}
\usepackage[T1]{fontenc} 
\usepackage[english]{babel} 
\usepackage{amsmath,amsfonts,amsthm,amssymb} 
\usepackage{algorithm}
\usepackage{enumitem}
\usepackage[noend]{algpseudocode}
\usepackage{cite,graphicx,float}
\usepackage{subcaption}
\usepackage{nccmath}
\usepackage{caption}
\usepackage{xcolor}
\usepackage[figurewithin=none]{caption}

\usepackage{fancyhdr} 
\newtheorem{theorem}{Theorem}
\newtheorem{lemma}{Lemma}

\def\bm#1{\mathbf{#1}}

\title{	\huge Optimal Transmission Using a Self-sustained Relay in a Full-Duplex MIMO System}

\author{Rafia Malik and Mai Vu\\
Department of Electrical and Computer Engineering, Tufts University, Medford, MA, USA\\
Email: rafia.malik@tufts.edu, mai.vu@tufts.edu} 

\date{\normalsize{March 17, 2017}} 

\begin{document}
\maketitle 
\begin{abstract}
This paper\footnote{\textcolor{red}{Full citation for published paper- R. Malik and M. Vu, "Optimal Transmission Using a Self-Sustained Relay in a Full-Duplex MIMO System," in IEEE Journal on Selected Areas in Communications, vol. 37, no. 2, pp. 374-390, Feb. 2019.}} investigates wireless information and power transfer in a full-duplex MIMO relay channel where the self-sustained relay harvests energy from both source transmit signal and self-interference signal to decode and forward source information to a destination. We formulate a new problem to jointly optimize power splitting at the relay and precoding design for both the source and relay transmissions. Using duality theory, we establish closed-form optimal primal solutions in terms of the dual variables, based on which we then design a customized and efficient primal-dual algorithm to maximize the achievable throughput. Numerical results demonstrate the rate gains from using multiple transmit and receive antennas in both information decoding and energy harvesting, and the significant benefit of harvesting energy from self-interference signals. We also extend our analysis to the case when channel state information is only available at receiving nodes and show how our algorithm can optimize the power splitting at the relay for it to remain self-sustained. Through analysis and simulation, we demonstrate that an optimal combination of non-uniform power splitting, variable power allocation, and self-interference power harvesting can effectively exploit a full-duplex MIMO system to achieve significant performance gains over existing uniform power splitting and half-duplex transmissions. 
\end{abstract}

\begin{IEEEkeywords}
Full-duplex MIMO, energy-harvesting, self-interference, precoding, power splitting, optimization
\end{IEEEkeywords}

\section{Introduction}
Far-field, radio frequency (RF) energy harvesting has recently garnered significant interest for communication systems with the prospect of simultaneous information and power transfer. Theoretical bounds showing the trade-off between energy and throughput have been studied in literature~\cite{Ulukus2015}. Due to the physical limitation of an energy harvesting circuit in its inability to process any information embedded in the RF signal, techniques like time switching and power splitting are seen in practice to divide the RF signal into separate parts for information decoding and energy harvesting~\cite{Krikidis2014}~\cite{Rui2013}. In general, power splitting is more efficient in utilizing the available signal power in any given time slot and therefore leads to reduced transmission delay and increased spectral efficiency compared to time switching. Uniform power splitting is commonly employed when developing dual purpose wireless power/information transfer models~\cite{Alouini2016}~\cite{Rui2013}. In this scheme, all antennas at the receiver split power using the same ratio between information decoding and energy harvesting parts. While such uniform power splitting among the antennas is simpler from an analysis and implementation point of view, it may not always be optimal. Recently, significant performance gains have been demonstrated for a half-duplex MIMO relay channel by directing the relay received power into beams and using non-uniform power splitting on these beams~\cite{Malik2018}.

Wireless charging of communicating nodes offers the inherent advantage of untethered mobility. Utility of wireless power transfer in self-sustained relays, in particular, has emerged as an interesting research avenue, where the relays use harvested power for information forwarding, rather than depleting their own power resources~\cite{Bruno2018}. Cooperation in communication networks, using relays, leads to performance enhancement by overcoming the effects of shadowing or by extending network coverage~\cite{Molisch2012}. With the increase in the number of connected devices, especially those with low power requirements, we can envision future networks employing relays capable of providing cooperation in terms of both information and power. In terms of wireless powered communication networks, having energy harvesting relay nodes can be particularly beneficial for distributed cooperation or multi-user energy transmission since unlike wired charging, the nature of wireless charging is broadcast, enabling both directed and opportunistic energy harvesting~\cite{Kim2016}.

Most wireless architectures in literature devise algorithms for relay channels with the half-duplex constraint~\cite{Alouini2016} and using single antenna links~\cite{Blum2016}\cite{Poor2016}\cite{Tran2014}. To cope with the growing demand for high throughput transmission, however, MIMO and full-duplex transmissions may offer a viable solution. MIMO links can exploit their degrees of freedom gain from the increased spatial dimensions for communication by using multiple transmit and receive antennas. By spatially multiplexing several data streams onto the channel, MIMO can offer enhanced throughput performance~\cite{Tse2005}. Prior works in joint information and energy transfer have considered systems with multiple antennas at the transmitting nodes and single antennas at the receiving nodes~\cite{Rui2015}\cite{Zhao2017}, or systems with multiple transmit/receive antennas only at the relay node~\cite{Mohammadi2015}. While such systems can provide throughput gains over single antenna transmission, a MIMO system with multiple antennas at all nodes can further enhance the system's performance not only in terms of the achievable rate, but also the harvested energy.

Full Duplex communication is another technology that can help boost throughput gains, since it offers increased spectral efficiency compared to half duplex communication by utilizing the entire time and bandwidth for two-way data transmission. However, tackling the performance degradation due to self-interference in full-duplex communication is one of the main challenges to its adoption. This self-interference deteriorates the performance more as the number of antennas is increased as is the case for MIMO. Two classes of techniques for self-interference cancellation exist in literature-- passive cancellation and active cancellation~\cite{Hanzo2015}. Passive cancellation can be achieved using directional isolation of the main lobes of the transmit/receive antenna array. Self-interference can also be passively reduced by increasing the path loss via shielding or by increasing the antenna separation~\cite{Sahai2014}. Active cancellation, on the other hand, uses a separate RF chain for self-interference mitigation~\cite{Katti2014}. It can employ training-based algorithms in the time domain, or exploit the increased degrees of freedom offered in the spatial domain by the antenna arrays of MIMO. One example of active cancellation is to use cancellation circuits together with DSP algorithms, relying on training and channel estimation for the self-interference channel~\cite{Katti2014}. RF impairments, however, hinder the complete elimination of self-interference, and some residual self interference remains in the baseband~\cite{Korpi2014}. Baseband residual self-interference leads to an increase in the noise floor, or equivalently corresponds to an SNR loss, and is therefore detrimental to the achievable rate performance; however, a recent combination of passive and active cancellation has been demonstrated to bring this residual self-interference down to the noise floor~\cite{Katti2014} \cite
{Valkama2016}. On the other hand, self-interference may prove beneficial in the RF domain, where the energy from the self-interference signal can be harvested, and consequently increase the power available to the node for subsequent transmissions~\cite{Kim2016}. 

It is not too far-fetched to envision a wireless system employing all the aforementioned technologies: wireless power transfer, MIMO and full-duplex communication in a cooperative relaying setting. One application could be drone-assisted cellular networks, where drones act as mobile base-stations (for instance in lieu of a failed base-station or to provide temporary coverage for a sports event). These drones have the capability of simultaneously performing wireless charging along with relaying users' cellphone signals to the main base station~\cite{Hossain2018}\cite{Kakishima2018}. Such a novel system inherently brings about new challenges, particularly in the design of transmit signals and the optimization for wireless power charging. 
\subsection*{Major Contributions}
In this work, we consider a MIMO communication system assisted by a full-duplex relay, where the relay is capable of harvesting energy. This model generalizes several existing problems considered in literature on Simultaneous Wireless Information and Power Transfer (SWIPT) systems by integrating both MIMO and full-duplex features, which to our knowledge is the first to do so. To focus on the benefits of MIMO and full duplex features, we consider the scenario in which the relay is self-sustained, that is, it has no power source of its own and hence relies solely on energy harvesting for its operations. The analysis and results obtained can be extended to the case where the relay has a power source and uses energy harvesting to supplement its power consumption.

Specifically, we formulate a novel optimization problem to maximize the throughput in a wireless powered two-hop full duplex MIMO relay channel, where all the nodes, source, relay and destination, are equipped with multiple antennas. Our formulation explicitly considers the effect of self-interference in full duplex communication. We show how using a combination of harvesting energy from the source signal along with the self-interference signal in the RF domain, and using active self-interference cancellation in the baseband, allows us to exploit performance gains of multiple antennas and full-duplex communications. Our analysis shows the effects of various MIMO channels including the self-interference channel on the resulting optimal power allocation and non-uniform power splitting ratios. Based on the Lagrange duality theory, we design an efficient primal-dual algorithm for jointly optimizing the source, relay precoders and the relay power splitting ratio. The algorithm allows us to demonstrate the significant performance gain of our model over the traditionally used uniform power splitting scheme, half-duplex MIMO transmission, and also full-duplex MIMO transmission without self-interference energy harvesting. We further extend our algorithm to the practical scenario where the channel state information is available at receiver nodes only. 

The main contributions of this work can be summarized as follows.
\begin{enumerate}
\item To our best knowledge, this is the first work to consider a truly MIMO setting for a full-duplex SWIPT relay channel, with multiple antennas at all nodes. While MIMO self-sustained relay has been considered in~\cite{Alouini2016}, it was for a half-duplex setting and employed uniform power splitting. Having multiple antennas at all nodes and a full-duplex relay introduces complicated coupling between all optimizing variables, namely the precoding matrices in both hops and relay power splitting ratios which can be non-uniform. 
\item This work shows the optimality of non-uniform power splitting for MIMO transmissions and demonstrates the performance gains in comparison to the traditional uniform power splitting~\cite{Alouini2016}\cite{Zhao2017}. While uniform power splitting can be optimal for systems with a single antenna at one or all nodes~\cite{Rui2013}, it is shown to be strictly sub-optimal for MIMO transmission~\cite{Malik2018}. Uniform power splitting, however, is simpler to implement in practice. Thus in addition to our main focus on non-uniform power splitting, we also propose an efficient algorithm for uniform power splitting for full-duplex relaying which is shown to be more efficient compared to existing methods employing grid searches~\cite{Alouini2016}.
\item We utilize RF self-interference of a full-duplex receiver towards increasing the harvested energy at the relay and also consider the power loss for interference cancellation circuitry required for information decoding in baseband. While the concept of self-energy recycling has been considered for a MISO system in~\cite{Rui2015}, unlike our work, this reference did not consider a truly full-duplex system in terms of information transmission, rather the strategy is for the relay to transmit information in the second hop while concurrently harvesting power in the first hop from the source signal which is only used for power transferring in this portion of the time slot. Our work is the first to demonstrate the potentially immense benefit of harvesting self-interference energy at radio frequencies in a full-duplex system, while still allowing this interference to be canceled for information decoding at baseband.
\item We use the concept of analog beamforming at the relay to rotate the received signal (energy beam), which provides additional degrees of freedom for simultaneously maximizing the transmission rate and the energy harvested from both the source signal and the self-interference signal arriving at the relay. The idea of analog beamforming at the relay was also proposed in~\cite{Malik2018} for half-duplex transmission, which helped decouple the rate optimization of the two hops and reduce the power allocation in second hop to simple water-filling. For full-duplex transmission, the self-interference signal strongly couples the two hops via the harvested energy constraint and significantly complicates the analysis of the optimization. 
\item We propose an efficient algorithm to solve the optimization problem to maximize the overall source to destination (S-D) transmission rate. The S-D throughput is affected by the power splitting ratios, the analog receive beamforming at the relay and the precoding designs at the source and relay, whose optimization is posed as a semi-definite programming (SDP) problem. Through analysis, we re-formulate this complicated problem into an equivalent simpler convex problem with scalar variables, and propose a computationally efficient algorithm which can perform multiple times faster than standard solvers~\cite{cvx} and prior sequential algorithms~\cite{Alouini2016}. 
\end{enumerate}
\begin{figure}[t]
\centering
\includegraphics[scale = 0.485]{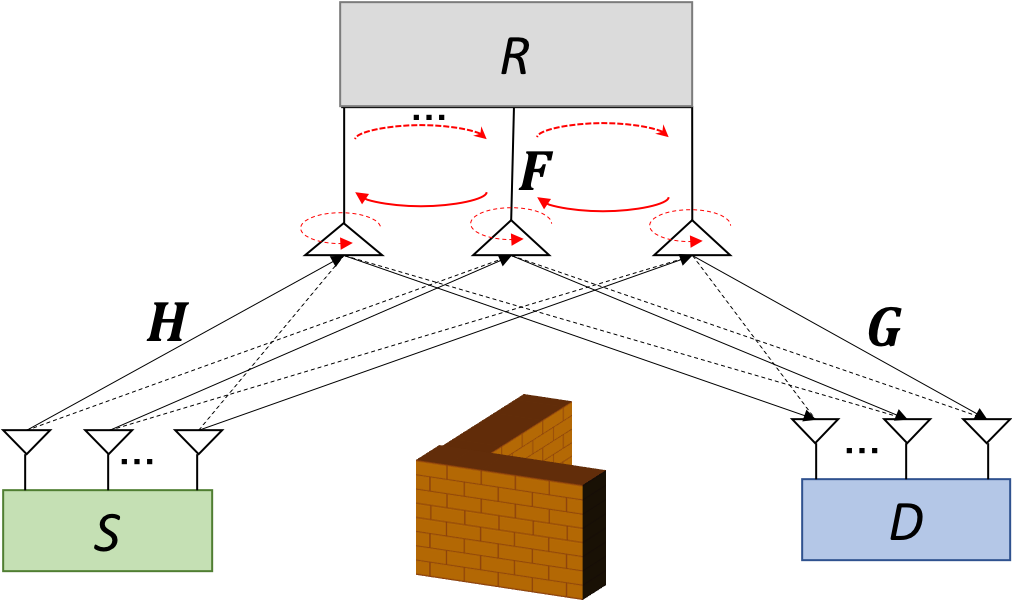}
\caption{MIMO two-hop relaying channel model}
\label{fig:channel_model}
\end{figure} 
\subsubsection*{Organization}The remainder of this paper is organized as follows. Section II presents the channel and signal models. Section III presents the formulation of the rate maximization problem and its dual problem characterization. Section IV presents the optimal primal solution for power splitting and allocation. Section V presents the primal-dual optimization algorithm, Section VI presents the optimization problem and its solution for the CSIR case and Section VII shows numerical results and analysis. Finally, Section VIII concludes this paper.

\subsubsection*{Notation} For a square matrix $\boldsymbol{X}$, $\text{tr}(\boldsymbol{X})$, $\lvert \boldsymbol{X} \rvert$, $\boldsymbol{X}^{-1}$, $\boldsymbol{X}^{\dagger}$ denotes the trace, determinant, inverse and pseudo-inverse respectively, and $\boldsymbol{X} \succeq 0$ means that $\boldsymbol{X}$ is a positive semi-definite matrix. For an arbitrary sized matrix, $\boldsymbol{Y}$,  $\boldsymbol{Y}^\ast$ denotes the Hermitian transpose, $\otimes$ denotes the Kronecker product, and $\textbf{diag}(y_1,...,y_N)$ denotes an $N\times N$ matrix with diagonal elements $y_1,...,y_N$. $\boldsymbol{I}$ denotes an identity matrix, and $\boldsymbol{0, 1}$ denote an all zeros vector and all ones vector respectively. The standard circularly symmetric complex Gaussian (CSCG) distribution is denoted by $\mathcal{CN}(\boldsymbol{0}, \boldsymbol{I})$, with mean $\boldsymbol{0}$ and covariance matrix $\boldsymbol{I}$. $\mathbb{C}^{k \times l}$ and $\mathbb{R}^{k \times l}$ denote the space of $k \times l$ matrices with complex and real entries respectively. Superscript $^\star$ denotes the optimal value for the corresponding variable.

\section{System Model}
\subsection{Channel Model}
Consider a full duplex, decode-and-forward (DF) MIMO relay communication channel, where the direct transmission link suffers from significant path loss and fading, such that the relay channel is always used for two-hop data transmission from source to destination as shown in Figure~\ref{fig:channel_model}. The source S, relay R and destination D are equipped with $N_s$, $N_r$ and $N_d$ antennas respectively, and the S-R and R-D fading channels are modeled by matrices $\textbf{H} \in \mathbb{C}^{N_r \times N_s}$ and $\textbf{G} \in \mathbb{C}^{N_d \times N_r}$ respectively. We employ the standard assumptions of quasi-static block fading channels and that each node can access perfect local channel state information on its receive links. For transmiters' CSI, we will consider both cases: (i) perfect CSI at the transmiters in order to reveal theoretical performance bounds, and (ii) no CSI at the transmiters for comparison with a practical scenario. For the case with perfect transmit CSI case, our analysis and results are independent of channel statistical models and are applicable to any; whereas for the case with no transmit CSI, we assume an i.i.d. Rayleigh fading channel model where the entries of the channel gain matrices, $\bm H$ and $\bm G$, are independent, identically distributed and circular symmetric complex Gaussian~\cite{Rappaport2009}.

\begin{figure}
\centering
\includegraphics[scale=0.5]{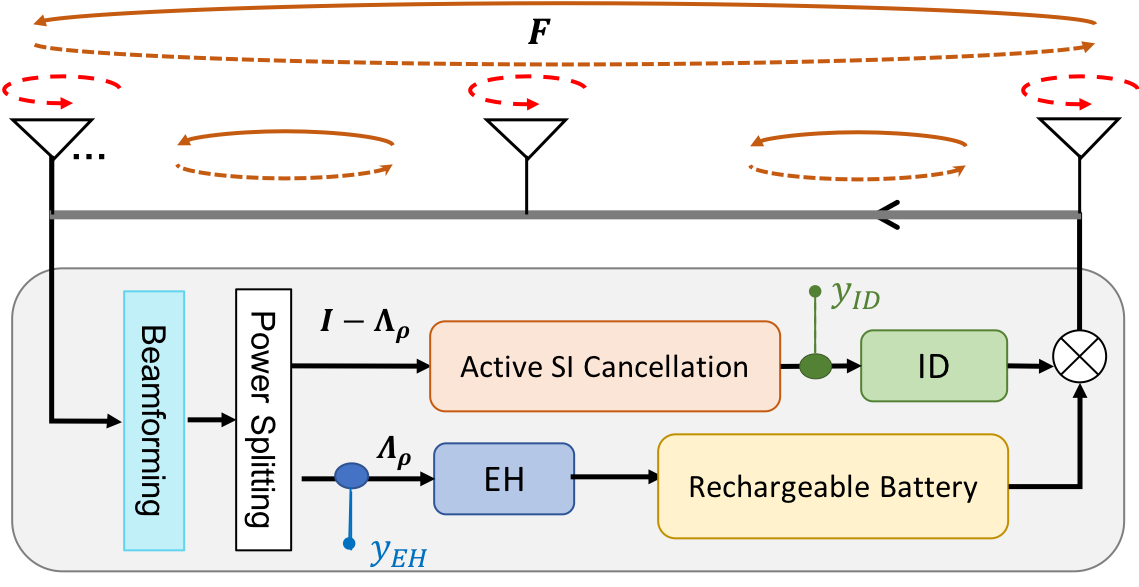}
\caption{Full-duplex MIMO relay model with self-interference}
\label{fig:relay_model}
\end{figure}

\subsection{Full-duplex MIMO relay and self interference models}
We consider a full-duplex MIMO relay as shown in Figure~\ref{fig:relay_model} with $N_r$ antennas used for simultaneous transmission and reception as successfully demonstrated in~\cite{Katti2014} (note that our model can also apply to the case of separate transmit and receive antennas as in~\cite{Sahai2014}). Similar to~\cite{Malik2018}, we introduce a receive beamforming matrix $\boldsymbol{Q_r}$ at the relay which performs analog beamforming in the RF domain to rotate the received signal and direct the received energy in beams to the power splitter, which acts on each received beam. The traditional power splitting per antenna is included as a special case. This received beamforming matrix provides more degrees of freedom in optimizing for both the transmission rate and harvested energy and is essential in decoupling the precoder design in the two hops and power splitting in later rate optimization. 

Having a full-duplex relay introduces self-interference which impacts both energy harvesting and signal decoding. The self-interference signal at the radio-frequency (RF) carries energy and is useful for energy harvesting~\cite{Rui2015}. To harvest this RF energy, we need to model the RF self-interference channel as seen by relay antennas. On the other hand, the same self-interference signal after going through both passive and active interference cancellation circuitry will impede information decoding at the baseband~\cite{Katti2013}. For this reason, we also need to model the baseband residual self-interference after cancellation. Note that self-interference at the radio frequency and at the baseband are different (see Figure~\ref{fig:relay_model} for the relay model of which we will discuss the signal details in Section II-C). Next we discuss in detail the models of self-interference channel at the radio frequency for energy harvesting and residual self-interference at the baseband for information decoding.

RF self-interference channel constitutes of a quasi-static internal self-interference component and an external time-varying component~\cite{Chen2018}. The internal subchannel depends on the physical isolation between transmit and receive chains/antennas,  while the external subchannel depends on the reflections from the surroundings. We adopt the experimental model in~\cite{Sabharwal2012} for the self-interference channel in the RF domain, $\bm F$, as a Rician fading channel with a strong line of sight component, that is, with a high K-factor. Unlike the separate antenna model in~\cite{Sabharwal2012}, however, our system uses a shared antenna model, whose simultaneous transmission and reception implies that the strong transmit signal from the transmitter chain would directly couple to the receiver chain~\cite{Korpi2016}. While the technique for physical isolation as a means of self-interference cancellation is different for the separate and shared antenna cases, the self-interference signal characteristics can be assumed to be similar~\cite{Xing2017}.

Next we discuss the residual self-interference (RSI) at the baseband after the self-interference signal goes through cancellation circuitry. In contrast to the RF self-interference, statistical characterization for the RSI in baseband still remains an open problem, especially for MIMO radios~\cite{Sabharwal2014}. The RSI amount is highly dependent on the cancellation methods implemented, especially active self-interference cancellation. Recent full-duplex MIMO experiments and measurements have demonstrated that the residual interference power is independent of both the transmit power and the number of antennas, consequently it can be considered as an increase in noise floor or equivalently as an SNR loss (see~\cite{Katti2014}, Figures 10,11) (note that the same result was also observed for experimental single antenna systems in~\cite{Katti2013}, Figure 7). We apply this result and assume that the RSI in baseband is an additional, independent noise source, leading to an SNR loss with respect to the achievable $R_{S-R}$ rate. Since the statistical characterization of RSI is still largely unknown, like previous works (~\cite{Katti2013}\cite{Katti2014}\cite{Tran2014} ) we assume this RSI to be zero mean, additive Gaussian noise. This assumption is reasonable since there are numerous sources of imperfections in the RF chain, and central limit theorem can be applied to this unpredictable (random) noise (RSI). 

Note that in the literature, there are models depicting the RSI noise power proportional to the transmit power~\cite{Sabharwal2012}\cite{Tran2014}, depending on the employed self-interference cancellation techniques. Nevertheless, the recent full-duplex MIMO experiment in~\cite{Katti2014} shows that significant reduction in RSI is now possible, and it can be decreased to be as low as 1 dB in baseband independent of the transmit power.  The same empirical results also show that the RSI power remains approximately constant as the number of RF chains (corresponding to antennas in a MIMO radio) is increased. For this reason, we model the RSI channel in baseband as an additional Gaussian noise with constant power, i.e., $\bm B \sim \mathcal{N} (\bm 0, \sigma_f^2 \bm I )$. However, in the numerical results, we also include results for the case when RSI in baseband is directly proportional to the relay transmit power $P_r$, such that $\sigma_f^2 = \alpha P_r^\beta$, with $\beta = 1$ which implies a linear increase in RSI with $P_r$.

The significant mitigation of residual self-interference in baseband to the noise floor level requires a combination of passive and active cancellation techniques. While the former does not consume any additional power, the latter makes use of an additional RF chain with analog to digital converters, RF attenuators etc, and these components consume power. With recent advances in semiconductor technology, low power RF components are readily available~\cite{Lee2014} and power-efficient techniques for self interference cancellation have been developed~\cite{Sayed2017}\cite{Tijani2017}. We therefore reasonably assume and then show via numerical results that provided sufficient transmit power from the source, the gain obtained from self-interference power harvesting in the RF domain at the relay would make up for the power consumed in active cancellation, such that the relay operation can remain self-sustained with no external power source requirement. Although not considered in this paper, our models,  analysis and results can be extended in a straightforward way to the case that the relay has its own power source. In either case, we will need to deduct an amount of power used for active self-interference cancellation from the RF harvested power before making it available for relay transmission.

The relay is a self-sustained node, employing a harvest-use policy using a set of two rechargeable batteries. The two batteries are used to store RF harvested energy and supply power using a time switching scheme. While one battery is used for powering transmission to the destination, the other is used to store the energy harvested during that time. Switching between the two batteries caters for the half-duplex constraint of energy transfer which prevents an energy storage device to charge and discharge at the same time to avoid thermal throttling~\cite{Rui2015}\cite{Luo2013}. This corresponds to a virtual harvest-use scheme for a battery model, where the harvested energy is utilized immediately and is not saved for future use~\cite{Poor2016}. Battery models for a harvest-store-use-scheme have been studied in the literature, where the stored energy is dispensed adaptively for future use~\cite{Ulukus2015}\cite{Blum2016}\cite{Poor2016}. Optimizing battery storage and usage then can be added as an extension to the optimization problem of the current system, but is out of the scope for this paper.

\subsection{Signal Model}
The relay harvests energy from both the source signal and the self-interference signal at the Energy Harvesting (EH) receiver and uses the harvested energy for transmitting the signal decoded at its Information Decoding (ID) receiver to the destination. Without loss of generality, we assume that both the energy harvesting and information decoding receivers are co-located at the relay and operate at the same frequency, the processing power for receive and transmit circuits at the relay is assumed to be negligible, except for the active cancellation circuitry (whose power consumption is taken into account in our formulation), and that both the batteries have sufficient energy initially. We make these assumptions considering the main motivation of this work which is to focus on the energy-harvesting at the relay related to full-duplex transmission and therefore assume negligible affect on our model from other circuit elements which are common in both full-duplex and half-duplex communications.

The received signal at the relay is fed into an analog beamforming matrix, which separates the signal into received beams. The received power in each beam is then divided for EH and ID, according to the power splitting ratio $\rho_i \in (0,1) \ \forall i \in [1,N_r]$, and we define $\boldsymbol{\Lambda_{\rho}} = \textbf{diag} (\rho_1 ... \rho_{N_r})$. We adopt the standard assumption that the power splitter is perfect and induces no noise in the RF signal. The input signal at the relay is as given below, where $\boldsymbol{y_{ID}}$ is the signal received at the information decoding receiver and $\boldsymbol{y_{EH}}$ is the signal received by the energy harvesting receiver as illustrated in Figure \ref{fig:relay_model}.   
\begin{align*}
&\boldsymbol{y_{ID}} = (\bm I - \boldsymbol{\Lambda_{\rho}})^{1/2} \boldsymbol{Q_r H x_s} + \boldsymbol{z_f}+ \boldsymbol{z_p} \tag{1a}\\
&\boldsymbol{y_{EH}} = (\boldsymbol{\Lambda_{\rho}})^{1/2} \boldsymbol{Q_r}(\boldsymbol{H x_s + F x_r}) \tag{1b}
\end{align*}
Here $\boldsymbol{x_s}$ is the signal vector transmitted from the source with the average transmit power constraint given as $\mathbb {E}\Big[\|\boldsymbol{x_s}\|_2^2\Big] = \text{tr}\Big(\boldsymbol{W_s}\Big) \leq P_s$, where $\boldsymbol{W_s} = \mathbb {E}\Big[ \boldsymbol{x_s x_s^\ast} \Big ]$ is the source covariance matrix. The noise term $\boldsymbol{z_p}$ with variance $\sigma_p^2$ denotes traditional receiver noise and $\boldsymbol{z_f}$ with variance $\sigma_F^2$ represents the effect of residual self interference in baseband as discussed earlier. We can then denote the total effective noise at the baseband as $\boldsymbol{z_r = z_f + z_p}$  which is distributed as $\mathbf{z_r} \sim \mathcal{CN}(\bm 0,(\sigma_p^2 + \sigma_F^2) \bm I)$. In (1b), $\boldsymbol{x_r}$ is the signal vector  transmitted from the relay and $\boldsymbol{F}$ is the RF self-interference channel for energy harvesting. Note that in signal models, (1a) and (1b), we have clearly distinguished the effects of self-interference on information decoding (at the baseband via added noise $\boldsymbol{z_f}$) and energy harvesting (at radio frequencies via the self-interference channel $\boldsymbol{F}$), in accordance with our discussion of full-duplex relay and self-interference models in Section II-B. The receiver combining matrix $\boldsymbol{Q_r} \in  \mathbb{C}^{N_r \times N_r}$ which appears in equations (1a), (1b) is a unitary matrix representing receive analog beamforming. This beamforming matrix at the relay will be optimized to simultaneously maximize the transmission rate and the total signal power forwarded to the power splitter for optimal energy harvesting as discussed in later optimization.

The EH receiver harvests energy in the RF domain. By the law of energy conservation, we assume that the total harvested RF-band energy is proportional to that of the received baseband equivalent signal, that is, $P_{h,BB} = \eta_c P_{h,RF}$, where $\eta_c$ represents the transducer efficiency for converting harvested energy in RF domain to DC electrical energy for battery charging. For notational convenience we use $P_h$ to denote the harvested power in the baseband.

We assume that on average the energy consumed by the relay for transmission is equal to the energy harvested to prevent energy-outages in the data transmission phase~\cite{Ulukus2015}. We consider a linear energy harvesting model and assume that the received power at the relay is constant over a single time slot duration. For the single relay device considered in our model, this assumption of linear EH is logical. However, in the case of multiple energy harvesting relays, non-linear energy harvesting should be considered where the conversion efficiency varies with the received signal power~\cite{Schober2015}\cite{Clerckx2018}. The transmission power for the relay, $P_r$, is directly proportional to the harvested power, $P_h$, such that $P_r = \eta P_h$ where $\eta \triangleq \eta_c \eta_d \in [0,1]$, and $\eta_d$ is the utilizing efficiency for battery discharging. Following the approach in~\cite{Rui2013} we assume $\eta = 1$ without loss of generality. The power available for relay transmission can therefore be written as given below.
\setcounter{equation}{1}
\begin{align*}
P_r =  \frac{E_h}{T_s} &= \frac{\eta}{T_s} \big ( \text{tr} (\boldsymbol{\Lambda_{\rho} Q_r H W_s H^\ast Q_r^\ast}) \\
& + \text{tr}(\boldsymbol{\Lambda_{\rho} Q_r F W_r F^\ast Q_r^\ast}) \big) - P_{IC} \tag{2}
\end{align*}
\setcounter{equation}{2}Here the first term contains the source transmit covariance $\boldsymbol{W}_s$ and represents the energy harvested from the source transmit signal. In the second term $\boldsymbol{W}_r = \mathbb {E}\Big[\|\boldsymbol{x_r}\|_2^2\Big]$ is the relay transmit covariance matrix which contributes to relay energy harvested via self-interference, and finally $P_{IC}$ represents the power consumed for active self-interference cancellation as discussed earlier. Following the convention in~\cite{Poor2016}, we assume a block (time slot) of unit duration, $T_s = 1$ and use the expressions for harvested energy and harvested power interchangeably throughout this paper. 

The regenerative relay employs a decode-forward multi-hop relaying scheme~\cite{Gamal2011}. The relay recovers the message received from the sender
in each block and re-transmits it in the following block. The signal received at the destination is as given below, where $\mathbf{z_d} \sim \mathcal{CN}(\bm 0,\sigma_d^2 \bm I)$ is the additive noise at the destination.
\begin{equation}
\boldsymbol{y_d = Gx_r + z_d}
\end{equation}
The average transmit power constraint on the relay is related to the harvested power in (2) as $\text{tr}\Big(\boldsymbol{W_r}\Big) \leq P_r$. The receiver then decodes on the signal received from the relay to recover information transmitted from the source.

\subsection{Achievable Rate}
An achievable rate for a multi-hop relay channel is given as\cite[p.~387]{Gamal2011}.
\begin{align*}
&R = \max_{p(x_s) p(x_r)} \min\{I(X_s; Y_{ID} \vert X_r), I (X_r; Y_d)\}\\
&\hphantom{R} = \min \{\max_{p(x_s)} R_{S-R}, \max_{p(x_r)} R_{R-D} \}
\end{align*}
where the second expression follows from application of the first expression to the considered two-hop cascaded S-R and R-D channel and $R_{S-R}$ and $R_{R-D}$ are achievable rates of the first and second hop, respectively. Using the signal model in (1) and (3) and assuming optimal Gaussian transmit signals, per-hop achievable rates can be written as functions of the source and relay transmit covariance matrices $\boldsymbol{W_s}$ and $\boldsymbol{W_r}$ as
\begin{align*}
&R_{S-R} = \text{log}_2 \left | \mathbf{I} + \frac{\mathbf{(I - \Lambda_{\rho})Q_rHW_sH^\ast Q_r^\ast}}{{\sigma_p^2 + \sigma_f^2}} \right |, \\
&R_{R-D} = \text{log}_2 \left | \mathbf{I} + \frac{\mathbf{GW_rG^\ast}}{{\sigma_d^2}}  \right | \tag{4a}
\end{align*}
The two transmit covariances must satisfy the power constraints mentioned earlier, that is, $\text{tr}\Big(\mathbf{W_s}\Big) \leq P_s$ and $\text{tr}\Big(\mathbf{W_r}\Big) \leq P_r$. The overall achievable rate for S-D transmission, in bits/s/Hz, is then given as 
\begin{align*}\label{csb}
&R\big (\boldsymbol{\Lambda_{\rho}},\mathbf{Q_r, W_s,W_r} \big ) = \text{min} \Bigg\{ \max R_{S-R}, \max R_{R-D}\Bigg \} \tag{4b}
\end{align*}
\setcounter{equation}{4}To summarize, the end-to-end throughput is the minimum of the maximum rates achievable in each hop. 

\section{Optimization Problem Formulation}
\subsection{Power Splitting and Precoding Design Optimization}
In this section, we formulate the optimization problem for optimal power splitting and precoding design, assuming perfect CSI at both the receiver and transmitter (the case of no CSIT and only CSIR is considered later in Section VI). The maximization of the achievable end-to-end transmission rate is formulated as an optimization problem given below, where (5c) and (5d) are the transmit power constraints at the source and relay, respectively, and (5e) is the harvested power at the relay.
\begin{align*}
(\text {P}):&\underset {\boldsymbol{\Lambda_{\rho}}, \boldsymbol {W}_{s},\boldsymbol {W}_{r}, \boldsymbol{Q_r}}{\max } ~ R, \tag{5}\\
\text{s.t.}~ &R \leq \text{log}_2 \left | \mathbf{I} + \frac{\mathbf{(I - \Lambda_{\rho})Q_rHW_sH^\ast Q_r^\ast}}{{\sigma_p^2 + \sigma_f^2}}  \right |, \tag{5a}\\
&R \leq \text{log}_2 \left | \mathbf{I} + \frac{\mathbf{GW_rG^\ast}}{{\sigma_d^2}}  \right |, \tag{5b}\\
&\text{tr} \left ({\boldsymbol {W}_{s}}\right ) \leq P_{s}, \tag{5c}\\
&\text{tr} \left ({\boldsymbol {W}_{r}}\right ) \leq P_{r}, \tag{5d}\\
&P_r = \eta \big ( \text{tr} (\boldsymbol{\Lambda_{\rho} Q_r H W_s H^\ast Q_r^\ast}) \\
&\hphantom{P_r = }+ \text{tr}(\boldsymbol{\Lambda_{\rho} Q_r F W_r F^\ast Q_r^\ast}) \big) - P_{IC}, \tag{5e}\\
&\boldsymbol {W}_{s} \succeq 0,\quad \boldsymbol {W}_{r} \succeq 0. \tag{5f}
\end{align*}
Here $R\left ({\boldsymbol{\Lambda_{\rho}}, \boldsymbol {W}_{s},\boldsymbol {W}_{r}, \boldsymbol{Q_r}}\right )$ is the end to end transmission rate, as defined in~(\ref{csb}).

Problem (P) is a semi-definite programming (SDP) problem and includes the self-interference components in the harvested power expression and the S-R rate. By solving for the beamforming directions at the source and relay, we show in Lemma 1 that this problem reduces to solving only for the power allocation in the precoders and the power splitting factors for energy harvesting at the relay, which is equivalent to a simpler problem (P-eq) given below.
\begin{align}\label{P-eq}
&(\text {P-eq}) : ~\underset {R,\boldsymbol{\rho, p, q}}{\max } ~ R \tag{6}\\
&\hphantom {(\text {P}) : ~}\text {s.t.}~ R \leq \sum_{i=1}^{K_1} \log_2 \Big (1+\frac{(1-\rho_i)p_i\lambda_{H,i}}{\sigma_p^2 + \sigma_f^2}\Big ) \tag{6a}\\
&\hphantom {(\text {P}) : ~\text {s.t.}~} R \leq \sum_{j=1}^{K_2} \log_2 \Big( 1+\frac{q_j\lambda_{G,j}}{\sigma_d^2} \Big ) \tag{6b} \\ 
&\hphantom {(\text {P}) : ~\text {s.t.}~} \sum_{i=1}^{K_1} p_i \leq P_s \tag{6c} \\
&\hphantom {(\text {P}) : ~\text {s.t.}~} \sum_{j=1}^{K_2} q_j \left( 1 - \sum_{k=1}^{K_2} \rho_k \boldsymbol{\tilde{F}}_{kj} \boldsymbol{\tilde{F}}_{kj}^\ast \right ) \leq \sum_{i = 1}^{K_1} \lambda_{H,i} p_i \rho_i - P_{\text{IC}} \tag{6d}
\end{align}
Implicit constraints not mentioned here are $p_i \geq 0 , \ 0 \leq \rho_i \leq 1 \ \forall i \in [1,K_1], \text{ and } q_j \geq 0 \ \forall j \in [1,K_2]$. Here $\boldsymbol{\rho}$ represents the vector for power splitting ratios across the receiver beams, $\boldsymbol{p}$ denotes the power allocated across the eigenmodes from the source to relay channel, and $\boldsymbol{q}$ is the power allocation vector for transmission in the second hop from relay to destination. The matrix $\boldsymbol{\tilde{F} \triangleq U_H^\ast F V_G}$ where $\boldsymbol{U_H}$ and $\boldsymbol{V_G}$ are obtained from the Singular Value Decomposition (SVD) of the S-R and R-D channel matrices as $\boldsymbol{H = U_H \Sigma_H V_H^\ast}$ and $\boldsymbol{G = U_G \Sigma_G V_G^\ast}$, respectively. 
\begin{lemma}\label{equality}
The formulated problem (P) can be reduced to the optimization problem (P-eq) by choosing the source transmit beamforming, the relay receive beamforming and relay transmit beamforming as the singular vectors of the S-R and R-D channels. Specifically, $\boldsymbol{Q_r}^\star = \boldsymbol{U_H}^\ast$, $\mathbf{W_s}^\star = \mathbf{V_H \Lambda_s V_H^\ast}$ and $\mathbf{W_r^\star = V_G \Lambda_r V_G^\ast}$.
\end{lemma}
\begin{proof} Here we provide the sketch of the proof, the full proof is given in Appendix~\ref{appendixA}.\\
It can be shown that the receiver combining matrix at the relay which simultaneously maximizes the first hop transmission rate and the relay's harvested energy from the source signal is given as $\boldsymbol{Q_r}^\star = \boldsymbol{U_H}^\ast$. The optimal source covariance matrix then has the form $\mathbf{W_s^\star = V_H \Lambda_s V_H^\ast}$, where $\mathbf{\Lambda_s} = \textbf{diag } (p_1 ... p_{K_1})$. For the second hop, the optimization is more complicated due to the coupling among the self-interference channel $\boldsymbol{F}$, relay transmit covariance $\boldsymbol{W_r}$, and the R-D channel $\boldsymbol{G}$ as seen in the expressions for the second-hop achievable rate and the harvested energy. Since the self-interference harvested energy expression has more degrees of freedom for maximization (namely both the power splitting factors and power allocation factors at the relay) than the second-hop rate (only relay power allocation factors), we choose the relay transmit beamforming directions to maximize the second-hop transmission rate. As such, for the R-D channel, $\mathbf{W_r^\star = V_G \Lambda_r V_G^\ast}$, with $\mathbf{\Lambda_r} = \textbf{diag }(q_1 ...q_{K_2})$. Here $K_1 \text{ and } K_2$ are the number of active channels corresponding to the non-zero singular values of the channel matrices $\bm H \text{ and } \bm G$, respectively. Using matrix and trace identities, and the optimal forms for $\boldsymbol{Q_r, W_s}$ and $\boldsymbol{W_r}$, the transmission power at the relay, $P_r$, in constraints (5d) and (5e) of problem (P) is then equivalent to $P_r = \text{tr} \left( \boldsymbol{\Lambda_{\rho} \Sigma_H \Lambda_s \Sigma_H^\ast}\right ) + \text{tr} \left(\boldsymbol{\Lambda_\rho \tilde{F} \Lambda_r \tilde{F}} \right)$, where $\boldsymbol{\tilde{F} \triangleq U_H^\ast F V_G}$, which can then be simplified to the expression in (6d). 
\end{proof}
In the above proof, we introduce the matrix $\tilde{\boldsymbol{F}} = \boldsymbol{U_H^\ast F V_G}$ which models the self-interference channel rotated by the receive beamforming matrix $\boldsymbol{U_H}$ and the transmit beamforming matrix $\boldsymbol{V_G}$ at the relay. This new matrix $\tilde{\boldsymbol{F}}$ can be viewed as the effective self-interference channel seen by the relay after transmit and receive beamforming. Since both beamforming matrices are unitary, representing analog beamforming, their rotations do not alter the total amount of energy present in the self-interference channel  but only direct the energy in beam directions towards and from the relay. The relay power splitting between harvesting and transmission can then be equivalently optimized given this rotated or effective self-interference channel.

The proposed forms for the transmit beamforming at the source and the receive and transmit beamforming at the relay imply that the maximum rate in the full-duplex two-hop MIMO channel with self-interference and energy harvesting at the relay is achieved through spatial multiplexing in each hop. Since the precoding directions are eigenvectors of the channel matrices, then with CSI at the transmitter and receivers, the only parameters to be optimized are the transmit power allocation in the two hops and the relay power splitting ratios. 

Based on Lemma~\ref{equality}, the corresponding transmission rates in the first and second hop, and the harvested power (relay transmission power) are equivalently written as
\begin{align*}
&R_1 = \sum_{i=1}^{K_1} \log_2 \Big (1+\frac{(1-\rho_i)p_i\lambda_{H,i}}{\sigma_p^2 + \sigma_f^2}\Big ), \\
&R_2 = \sum_{i=1}^{K_2} \log_2 \Big( 1+\frac{q_i\lambda_{G,i}}{\sigma_d^2} \Big )\\
&P_r = \sum_{i = 1}^{K_1} \lambda_{H,i} p_i \rho_i + \sum_{k=1}^{K_2} \sum_{j=1}^{K_2} \rho_k q_j \boldsymbol{\tilde{F}}_{kj} \boldsymbol{\tilde{F}}_{kj}^\ast - P_{IC} \tag{7}
\end{align*}

\begin{lemma}\label{convexity}
The optimization problem (P-eq) is jointly convex in the optimizing variables, $R, \boldsymbol{p, q , \rho}$. Thus strong duality holds for this problem.
\end{lemma}
\begin{proof}
While it is straightforward to see that the objective function is affine and convex, constraint (6b) is convex and constraint (6c) is affine in each of the optimizing variables $q_i$ and $p_i$ respectively, however, joint convexity in the optimizing variables $(p_i, q_i, \rho_i)$ is required for the problem to be a convex optimization and needs to be established.  

To this end, we decompose the constraint functions in (6a) and (6d) into composition functions, and upon evaluating the Hessian and eigenvalues, show that the constraint functions for both these constraints are in fact quasiconvex for the domain of the optimization problem, with $p_i, q_i, \rho_i \geq 0$, where the sublevel sets of the constraint functions are convex. The problem then becomes maximization of a convex function, over convex sets and convex sub-level sets, and is therefore a convex optimization problem. Since there exists a strictly feasible point where the covariance matrices are scaled identity matrices, that is, $(\boldsymbol{W_s, W_r}) = (\gamma \boldsymbol{I}, \kappa \boldsymbol{I})$, with $\gamma$ and $\kappa$ as scaling coefficients, Slater's condition for constraint qualification is satisfied and strong duality holds. Details of the proof are included in Appendix~\ref{appendixB}. 
\end{proof}

Lemma~\ref{convexity} establishes that strong duality holds, which implies that the duality gap is zero. While generic convex optimization solvers~\cite{Boyd2004} such as~\cite{cvx} can then be used to solve this problem, they may not be efficient and will not reveal problem structure or additional insights into the optimal results. Next, we formulate a Lagrange Dual problem in order to analytically characterize the form of the optimal primal solution.

\subsection{Dual Problem Formulation}
The Lagrangian for (P-eq) is as given below
\begin{align*}
\mathcal{L}(R,\boldsymbol{\rho},\bm p,\bm q,\alpha,\beta,\nu,\mu) &= R - \alpha(R-R_1) - \beta(R-R_2) \\
&- \nu\big (\sum_{i=1}^{K_1} p_i - P_s\big ) - \mu\big (\sum_{j=1}^{K_2} q_j - P_r \big ),
\end{align*}
where $\alpha, \beta, \nu$ and $\mu$ are the dual variables corresponding to the constraints (6a) - (6d) respectively in (P-eq). Since strong duality holds, the primal and dual variables can be solved for as a primal-dual pair, $(p_i^\star, q_i^\star, \rho_i^\star,\alpha^\star, \beta^\star, \nu^\star, \mu^\star)$. Since $\mathcal{L}$ is a linear function of $R$ and hence differentiable with respect to $R$, we use the optimality condition and set $\nabla_R \mathcal{L} = 0$, which gives us $1-\alpha-\beta = 0 \implies \beta = 1 - \alpha$. This is also intuitive, since at any point, the transmission rate, $R$, is equal to the minimum of the rate of the two hops, which is equal to either $R_1$ or $R_2$ or is equal to both when $R^\star = R_1 = R_2$. Substituting $\beta = 1-\alpha$, we can rewrite the Lagrangian as
\begin{align*}
&\mathcal{L}(\boldsymbol{\rho},\bm p,\bm q,\alpha,\nu,\mu) = \alpha R_1 + (1-\alpha)R_2 - \nu\left (\sum_{i=1}^{K_1} p_i - P_s\right ) \\
&-  \mu\left (\sum_{j=1}^{K_2} q_j - \sum_{k=1}^{K_2} \sum_{j=1}^{K_2} \rho_k q_j \boldsymbol{\tilde{F}}_{kj} \boldsymbol{\tilde{F}}_{kj}^\ast - \sum_{i = 1}^{K_1} \lambda_{H,i} p_i \rho_i - P_{\text{IC}}\right )
\end{align*}

We now define the Dual problem as follows, such that minimization of the dual objective function is equivalent to the maximization of the primal objective function.
\begin{align*}
\text{P-Dual} &: \min_{\alpha, \nu, \mu} g(\alpha, \nu, \mu)\\
&\text{where:}\ g(\alpha, \nu, \mu) = \max_{R, \boldsymbol{\rho, p, q}} \mathcal{L} (R,\boldsymbol{\rho, p, q}, \alpha, \nu, \mu) \tag{8}
\end{align*}
In Section IV next, we solve for the optimal primal variables which maximize the Lagrangian to obtain the dual function. This then provides a basis for a  Primal-Dual algorithm to solve the optimization problem (P-eq) which we will discuss in Section V. 

\section{Optimal Primal Solution} 
In establishing the dual function, we solve in closed-form for the optimal primal variables in terms of the dual variables. We find closed form expressions of the primal variables for problem (P-eq) which is formulated for the case of non-uniform power splitting scheme. As an extension, we also solve for the primal variables for the case of uniform power splitting scheme which will be used for later performance comparison.

\subsection{Non-Uniform Power Splitting}
\begin{theorem}
The optimal power allocation, $\bm p$ and $\bm q$, in the first and second hop respectively, and the optimal power splitting ratio, $\boldsymbol{\rho}$, can be obtained in terms of the dual variables as
\begin{align*}
&p_i^\star = \Bigg ( \frac{\alpha}{\nu - \mu \rho_i \lambda_{H,i}} - \frac{\sigma_p^2 + \sigma_f^2}{\lambda_{H,i} (1 - \rho_i)} \Bigg )^+ \tag{9a}\\
&q_j^\star = \Bigg ( \frac{1-\alpha}{\mu (1 - \sum_{k=1}^{K_2} \rho_k \boldsymbol{\tilde{F}}_{kj} \boldsymbol{\tilde{F}}_{kj}^\ast )} - \frac{\sigma_d^2}{\lambda_{G,j}} \Bigg )^+ \tag{9b}\\
&\rho_k^\star = \Bigg ( 1 + \frac{\sigma_p^2 + \sigma_f^2}{\lambda_{H,k} p_k}- \frac{\alpha}{\mu ( \lambda_{H,k} p_k + \sum_{j=1}^{K_2} q_j \boldsymbol{\tilde{F}}_{kj} \boldsymbol{\tilde{F}}_{kj}^\ast )} \Bigg )_0^1 \tag{10}
\end{align*}
\end{theorem}
\begin{proof}
Obtained through KKT conditions by setting $\nabla_{p_i} \mathcal{L} = 0$, $\nabla_{q_i} \mathcal{L} = 0$ and $\nabla \mathcal{L}_{\rho_i} = 0$, respectively. Details are included in Appendix~\ref{appendixC}. Here $(x)^+ = \max(x,0)$, and $(x)_0^1 = \max (\min (x,1),0)$ to ensure implicit constraints $p_i, q_i \geq 0$ and $0 \leq \rho_i \leq 1$ respectively. 
\end{proof}

Observing the expressions for power allocation in both hops, $p_i$ in (9a) and $q_i$ in (9b), the optimal power levels are varied according to the channel eigenmodes and power splitting ratios in a water-filling fashion where multiple and varying water "levels" depend on the dual variables and related channel eigenvalues. The power splitting ratios are also varied according to the S-R and self-interference channel eigenvalues to achieve a tradeoff between information decoding and energy harvesting. 

For robust implementation in our algorithm later, it is useful to express $p_i$ in a standard waterfilling form as
\setcounter{equation}{10}
\begin{equation}
p_i^\star = \frac{\alpha}{\nu - \mu \lambda_{H,i} \rho_i} \underbrace{\Bigg(1 - \frac{(\nu - \mu \rho_i \lambda_{H,i})(\sigma_p^2 + \sigma_f^2)}{\alpha (1 - \rho_i) \lambda_{H,i}} \Bigg )^+}_{\tilde{p}_i}
\end{equation}
In this form, the expression for $\tilde{p}_i$ in the bracket is standard water-filling with a constant "water level", and $p_i^\star$ is then obtained from $\tilde{p}_i$ via an index-specific scaling factor. From (11) we can see that, for bounded power allocation $p_i^\star$, we require $\nu > \mu \lambda_{H,i} \rho_i$, and $\rho_i < 1 \ \forall i$. To perform waterfilling for allocating $\tilde{p}_i$, we use a dual-variable dependent sum power constraint as $\sum_{i=1}^{N_1} \tilde{p}_i = \frac{1}{\alpha} (\nu P_s - \mu \sum_{i=1}^{N_1} p_i \lambda_{H,i} \rho_i)$ in which the term $\sum_{i=1}^{N_1} p_i \lambda_{H,i} \rho_i$ represents the amount of power harvested from the source signal only. Similarly for $q_i$, we have
\begin{equation}
\medmath{q_i^\star = \frac{1 - \alpha}{\mu (1 - \sum_{k=1}^{K_2} \rho_k \boldsymbol{\tilde{F}}_{kj} \boldsymbol{\tilde{F}}_{kj}^\ast )}  \underbrace{\Bigg(1 - \frac{\mu \sigma_d^2 (1 - \sum_{k=1}^{K_2} \rho_k \boldsymbol{\tilde{F}}_{kj} \boldsymbol{\tilde{F}}_{kj}^\ast )} {(1 - \alpha) \lambda_{G,i}} \Bigg )^+}_{\tilde{q}_i}}
\end{equation}
Here $\tilde{q}_i$ is obtained through conventional waterfilling where the sum power constraint for $\tilde{q}_i$ can be obtained directly from (6d) in (P-eq) as $\sum_{i=1}^{K_2} \tilde{q}_i = \frac{\mu }{1 - \alpha} \Big [\sum_{i = 1}^{K_1} \lambda_{H,i} p_i \rho_i - P_{\text{IC}} \Big ]$.

\subsection{Uniform Power Splitting}
The expression for $\rho_i^\star$ in (10) is index-dependent on the specific eigenvalues of the S-R and self-interference channels, therefore the optimal PS ratio is non-uniform. Uniform power splitting, however, has often been adopted in literature to simplify the development of power splitting algorithms. Therefore, for comparison, we also solve for the optimal primal solutions with uniform power splitting ratio, which is constant across all antennas at the relay, that is, $\rho_i = \rho \ \forall i \in [1,N_r]$.

\begin{theorem}
The optimal power allocation, $\bm p$ and $\bm q$, in the first and second hop respectively, and the uniform power splitting ratio, $\rho^\star$, can be obtained in terms of the dual variables as
 \begin{align*}
 &p_i^\star = \Bigg ( \frac{\alpha}{\nu - \mu \rho \lambda_{H,i}} - \frac{\sigma_p^2 + \sigma_f^2}{(1-\rho)\lambda_{H,i}} \Bigg )^+ \\
 &q_j^\star = \Bigg ( \frac{1-\alpha}{\mu (1 - \sum_{k=1}^{K_2} \rho \boldsymbol{\tilde{F}}_{kj} \boldsymbol{\tilde{F}}_{kj}^\ast )} - \frac{\sigma_d^2}{\lambda_{G,j}} \Bigg )^+
 \end{align*}
The optimal PS ratio, $\rho^\star$ is obtained by solving the equation of the form $f(\rho) = c$ where $f(\rho)$ is a scalar function of $\rho$ and $c$ is a constant, and is given as
\begin{align}
&\medmath{\sum_{i=1}^{K_1} \frac{\alpha p_i \lambda_{H,i}}{\sigma_p^2+ \sigma_f^2 + (1-\rho) p_i \lambda_{H,i}} =  \mu \Big ( \sum_{j=1}^{K_2} \sum_{k=1}^{K_2} q_j \boldsymbol{\tilde{F}}_{kj} \boldsymbol{\tilde{F}}_{kj}^\ast + \sum_{i = 1}^{K_1} \lambda_{H,i} p_i \Big )}
\end{align}
\end{theorem}
\begin{proof}
Obtained through KKT conditions similar to the the non-uniform power splitting case, but by setting $\nabla_{\rho} \mathcal{L} = 0$, with scalar $\rho$ in this case . Details for power allocation proof are omitted to avoid redundancy. 
\end{proof}
\begin{lemma}
Solving for the optimal $\rho$ from (13) can be done efficiently via a bisection algorithm.
\end{lemma}
\begin{proof}
For $\rho^\star$, the function $f(\rho)$ is decreasing in $(1- \rho)$ and increasing in $\rho$, and is therefore monotonous in $\rho$, thus the solution for $\rho^\star$ can be obtained using a bisection algorithm.
\end{proof}
The algorithm finds the optimal $\rho^\star$ by searching for the root of the equation $g(\rho)= f(\rho) - c = 0$ in the interval (0,1) by repeatedly bisecting the interval for which the values of $g(\rho)$ have opposite signs at the two end points. That is, the interval $(l,u)$ brackets the root, if $g(l)g(u) < 0$.

\section{A Primal-Dual Algorithm}
\begin{algorithm}[t]
\caption{Solution for Rate Optimization Problem}
\textbf{Given:} Distance values $d_{sr}, d_{rd}, d_{sd}$. Channel Matrices $\boldsymbol{H, G, F}$. Precision $\epsilon_0$ and boundary value $\epsilon > 0$\\
\textbf{Initialize:} Dual variables $\alpha, \nu, \mu$. Primal variables $\boldsymbol{p, \rho}$. Ellipsoid shape matrix $P$.\\
\textbf{Begin Algorithm}
\begin{itemize}[leftmargin=*]
\item \textit{Calculate PS Ratio - Non-uniform Power Splitting}\\
Use the closed form expression in (10) to find $\rho_i^\star$. If $\rho_i^\star$ $\notin$ (0,1), use boundary conditions as follows, \\
$\circ$ If $\rho_i < 0 \implies \text{set } \rho_i^\star = \epsilon$ \\
$\circ$ Elseif $\rho_i > 1 \implies \text{set } \rho_i^\star = 1 - \epsilon$
\item \textit{Calculate PS Ratio - Uniform Power Splitting}\\
Use the bisection method (Lemma 2) to solve $f(\rho) = c$. If $\rho^\star$ $\notin$ (0,1), use the following boundary conditions: \\
$\circ$ If $R(\epsilon) \geq R(1 - \epsilon)$ $\implies \rho^\star = \epsilon$\\
$\circ$ Else $\rho^\star = 1 - \epsilon$
\item \textit{Calculate power allocation in first and second hop}\\
$\circ$ Using (11) calculate $P_{h1} = \sum_{i=1}^{N_1} p_i \lambda_{H,i} \rho_i$ and waterfill to find $\tilde{p}_i$. Find $p_i^\star$ from $\tilde{p}_i$\\
$\circ$ Based on (12) find $\tilde{q}_i$ using waterfilling. Obtain $q_i^\star$ from $\tilde{q}_i$. 
\item \textit{Check dual function value}
\begin{align*}
g(\alpha, \nu, \mu) &= \alpha R_1 + (1 - \alpha) R_2 - \nu (\sum_{i=1}^{N_1} p_i^\star - P_s) \\
&-\mu (\sum_{i=1}^{N_2} q_i^\star - P_h)
\end{align*}
$\circ$ If dual variables converge with required precision, \textbf{stop}\\
$\circ$ Else, find subgradients in (14), update dual-variables using ellipsoid method, \textbf{continue}
\end{itemize}
\textbf{End Algorithm}
\end{algorithm}
Theorems 1 and 2 show in closed form how the optimal primal variables are functions of the dual variables, but also show that they are interdependent on other primal variables. Thus in order to reach the optimal solutions, we need to to solve for the dual variables then the primal variables in an iterative fashion. To achieve this, we design a primal-dual algorithm which iteratively updates the dual and the primal variables until reaching a target accuracy.

To update the dual variables, we use the optimal primal solutions stated in Theorems 1 and 2 to obtain the dual function, $g(\alpha, \nu, \mu)$, as given in (8). The problem then becomes minimizing this dual function in terms of the dual variables. We can use a sub-gradient based method to solve for the dual minimization problem. The sub-gradient terms for the dual variables $(\alpha, \nu, \mu)$ are as given below
\begin{align*}
&\Delta \alpha = \sum_{i=1}^{K_1} \log_2 \left (1+\frac{(1-\rho_i)p_i\lambda_{H,i}}{\sigma_p^2 + \sigma_f^2}\right ) - \sum_{i=1}^{K_2} \log_2 \left( 1+\frac{q_i\lambda_{G,i}}{\sigma_d^2} \right ) \\
&\Delta \nu = \sum_{i=1}^{K_1} p_i - P_s\\
&\Delta \mu = \sum_{j=1}^{K_2} q_j \left( 1 - \sum_{k=1}^{K_2} \rho_k \boldsymbol{\tilde{F}}_{kj} \boldsymbol{\tilde{F}}_{kj}^\ast \right )- \sum_{i = 1}^{N_1} \lambda_{H,i} p_i \rho_i \tag{14}
\end{align*}
We then design a primal-dual algorithm which iteratively updates the dual and primal variables. The dual variables updates are based on the shallow-cut ellipsoid method using the sub-gradient expressions in (14), and the primal variables updates are based on results in Theorems 1 or 2 (depending on the case of non-uniform or uniform power splitting). 

The proposed primal-dual algorithm works as follows. First we need to initialize the variables. From the expressions for $p_i^\star$, $q_i^\star$ and $\rho_i^\star$ in (9) and (10), we see that they are interdependent. Therefore we initialize the algorithm with $\alpha_0 \geq 0, \nu_0 \geq 0, \mu_0 \geq 0, p_{i,0} \geq 0, q_{i,0} \geq 0$, and $0 \leq \rho_{i,0} \leq 1$. Next, we use the expressions in Theorem 1 and 2, for the non-uniform and uniform power splitting scheme respectively, to obtain values for the primal variables based on the current values for the dual variables. We use the modified waterfilling expressions in (11) and (12) to update the values for power allocation in the first and second hop. For the power splitting ratio, we use the closed form equation in (10) for non-uniform power splitting, and the bisection algorithm in Lemma 2 for the uniform power splitting scheme, respectively. 
\begin{figure}[t]
\centering
\includegraphics[scale = 0.5]{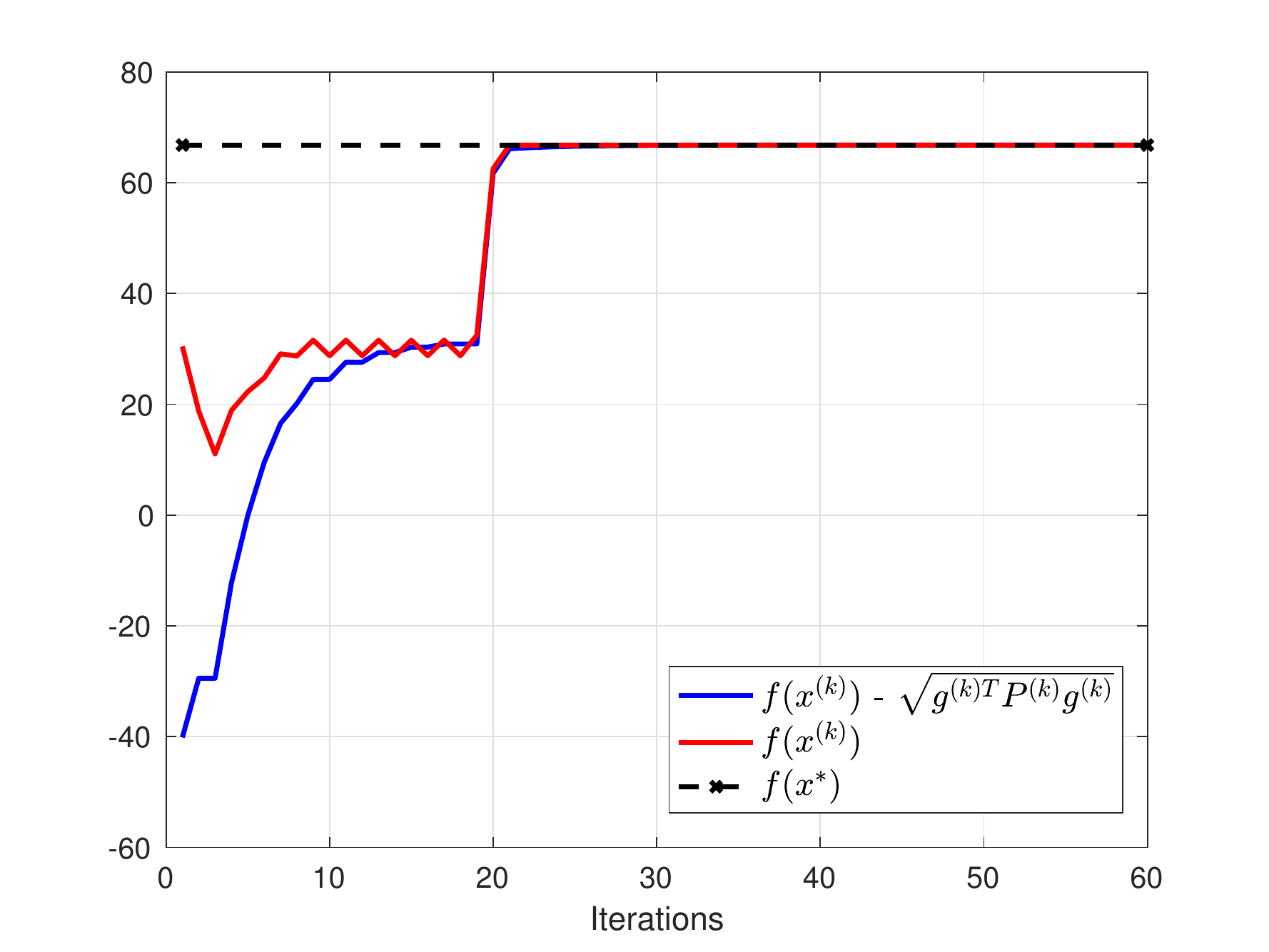}
\caption{A typical example of convergence to the optimal solution, $f(x^\star)$, with $x^\star = \{\alpha^\star, \nu^\star, \mu^\star \}$}
\label{convergence}
\end{figure}
After computing the primal variables, the dual variables are updated based on their respective sub-gradients given in (14) according to the shallow-cut ellipsoid algorithm~\cite{Boyd2004}. The sub-gradient in the ellipsoid algorithm is calculated at the ellipsoid center, $x = (\alpha, \nu, \mu)$, to reach the minimum volume ellipsoid containing the minimizing point for the dual-function $g(\alpha, \nu, \mu)$. For each iteration, the primal variables are updated according to the expressions in Theorem 1 and 2, and a value for $g(\alpha, \nu, \mu)$ is calculated. Since the ellipsoid algorithm is not a descent method, we keep track of the best point for $g(\alpha, \nu, \mu)$ at each update, however it uses modest storage and computation per step, $O(n^2)$, where $n$ is the number of variables. 

The algorithm thus iteratively updates the primal variables using closed-form equations with current values of the dual variables, then updates the dual variables using a sub-gradient method via the ellipsoid algorithm. These primal-dual update steps are repeated until the desired level of precision is reached for the stopping criterion, which is the minimum volume of the ellipsoid in our algorithm. As the optimal values for the dual variables are reached using the sub-gradient algorithm, the values for the primal variables also converge to their respective optimal values by strong duality (see Lemma~\ref{convexity}). Detailed steps are described in Algorithm 1. Figure~\ref{convergence} shows the convergence of the proposed primal-dual algorithm to the optimal solution.

\section{Channel State Information at receiver(s) only}
Our formulation, analysis and algorithms in previous sections have assumed the availability of local Channel State Information (CSI) at all nodes, including transmitting and receiving nodes. Obtaining CSI at transmitters, while possible in a time division duplex (TDD) system, can sometime be challenging in a wireless environment with fast fading or in frequency division duplex (FDD) systems~\cite{Tse2005}. We therefore extend the problem to the case where we have instantaneous CSI at the receiver (CSIR) nodes only. Without channel knowledge at the transmitter, it is optimal to allocate equal transmit power at the source and relay, such that $p_i = p = P_s/N_s \ \forall i$, and $q_i = q = P_r/N_r \ \forall i$~\cite{Goldsmith2007}. The optimization problem for the CSIR only case then reduces to just finding the optimal power splitting ratio ${\rho_i}$ between information decoding and energy harvesting at the relay. This problem can be posed as follows.

\begin{align}\label{P2}
&(\text {P2}) : ~\underset {R,\boldsymbol{\rho}}{\max } ~ R \tag{15}\\
&\hphantom {(\text {P})}\text {s.t.}~ R \leq \sum_{i=1}^{N_s} \log_2 \Big (1+\frac{(1-\rho_i)s P_s \lambda_{H,i}}{\sigma_p^2 + \sigma_f^2}\Big ) \tag{15a}\\
&\hphantom {(\text {P})\text {s.t.}~} R \leq \sum_{i=1}^{N_r} \log_2 \Big( 1+\frac{r P_r \lambda_{G,i}}{\sigma_d^2}\Big ) \tag{15b}\\
&\hphantom {(\text {P})\text {s.t.}~} P_r \leq \sum_{i=1}^{N_s} \rho_i s P_s \lambda_{H,i} + \sum_{i=1}^{N_r} \sum_{k=1}^{N_r} r P_r \rho_k \boldsymbol{\tilde{F}}_{ki} \boldsymbol{\tilde{F}}_{ki}^\ast \tag{15c}
\end{align}
Here scalar constants $r$ and $s$ are defined as $r \triangleq 1/N_r$ and $s \triangleq 1/N_s$. This optimization only concerns the relay, which as a receiving node has the knowledge of the S-R channel $\boldsymbol{H}$ and the self-interference channel $\boldsymbol{F}$ (active self-interference cancellation requires estimation of $\boldsymbol{F}$ at the relay~\cite{Katti2013}\cite{Katti2014}). We further assume that for power splitting at the relay, the relay node has knowledge of the eigenvalues for the R-D channel.  This assumption is reasonable since for a fading MIMO channel, the channel eigenmodes change much slower than the small scale fading parameters and can be obtained via feedback or estimation~\cite{Jafar2005}. 

Problem (P2) is simpler than (P-eq), since there are no power allocation variables at the transmit nodes as equal power is allocated for the transmitted signal from both the source and the relay. It is also more straightforward to see that (P2) is jointly convex in both $R$ and $\rho_i$. Constraint (15a) is linear in $R$ and convex in $\rho_i$, constraint (15b) is also linear in $R$, while constraint (15c) is affine and convex in $\rho_i$.

The Lagrangian for this case, with one primal variable vector $\boldsymbol{\rho}$, then becomes
\begin{align*}
&\mathcal{L}(\alpha, \mu, \boldsymbol{\rho}) = \alpha \sum_{i=1}^{N_s} \log_2 \Big (1+\frac{(1-\rho_i)s P_s \lambda_{H,i}}{\sigma_p^2 + \sigma_f^2} \Big ) \\
&+ \beta \sum_{i=1}^{N_r} \log_2 \Big( 1 + \frac{r P_r \lambda_{G,i}}{\sigma_d^2}\Big ) \\
&- \mu \Big (P_r - \sum_{i=1}^{N_s} s P_s \lambda_{H,i} \rho_i - \sum_{i=1}^{N_r} \sum_{k=1}^{N_r} r P_r \rho_k \boldsymbol{\tilde{F}}_{ki} \boldsymbol{\tilde{F}}_{ki}^\ast \Big ) \tag{16}
\end{align*}
where $\alpha$ and $\beta = 1 - \alpha$ are the dual variables associated with the rate constraints, (15a) and (15b), and $\mu$ is the dual variable corresponding to the relay power constraint (15c). We proceed with a similar approach as in Section IV to obtain the optimal power splitting ratio solution in terms of the dual variables as in Theorem 3.

\begin{theorem}
Non-uniform power splitting scheme is optimal even for the case with CSIR only, and the optimum power splitting ratio, $\rho_i^\star$, is given in closed form as
\begin{align*}
&\rho_i^\star = \Bigg ( 1 + \frac{\sigma_p^2 + \sigma_f^2}{\lambda_{H,i} s P_s}- \frac{\alpha}{\mu ( \lambda_{H,i} s P_s + \sum_{j=1}^{N_r} r P_r \boldsymbol{\tilde{F}}_{ij} \boldsymbol{\tilde{F}}_{ij}^\ast )} \Bigg )_0^1 \tag{17}
\end{align*}
\end{theorem}
\begin{proof}
Directly from KKT conditions by setting $\nabla \mathcal{L}_{\rho_i} = 0$ from (16).
\end{proof}  

Results of Theorem 3 allow us to design another primal-dual algorithm which uses (17) to update the primal variable's value and a sub-gradient method to update the dual variables. The sub-gradients, $(\Delta \alpha, \Delta \mu)$, are computed similar to (14). The Algorithm for CSIR-only case is omitted for brevity. However, the steps are similar to those of Algorithm 1 by replacing the appropriate closed form equations, with the exception of power allocation steps since for the CSIR-only case, there is equal power allocation with $p=P_s/N_s$ and $q = P_r/N_r$, unlike the optimal waterfilling power allocation solution in Algorithm 1.

\section{Numerical results and analysis}
\begin{figure}[t]
\begin{center}
\includegraphics[scale=0.5]{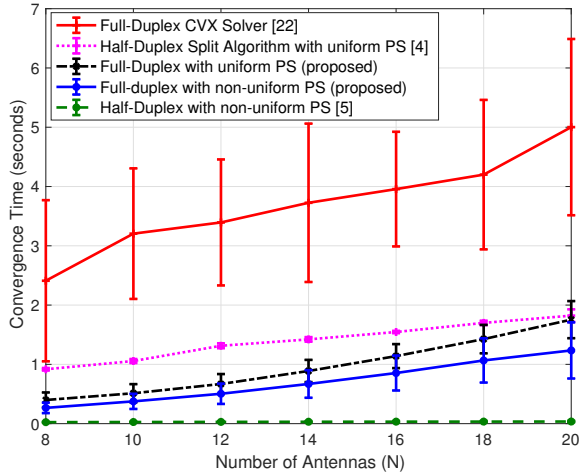}
\caption{Convergence comparison between the proposed algorithm for uniform and non-uniform power splitting, standard solver~\cite{cvx} and iterative algorithm~\cite{Alouini2016} for a 2xNx2 MIMO system with $P_s = 30$dBm}
\label{converg}
\end{center}
\end{figure} 
In this section, we analyze the performance of our proposed algorithms and compare with existing SWIPT techniques including half duplex and uniform power splitting. For simulations, we use path loss exponent $\gamma = 3.2$, and the distances between S-R and S-D respectively as $d_{sr} = 2\text{m and }d_{sd} = 10\text{m}$. At WiFi frequency 2.4 GHz for example, for a Uniform Linear Array (ULA) configuration with antenna elements separated by a distance $\lambda/2$, where $\lambda$ is the carrier wavelength, and reference distance $d_0 = 0.1$m (Fraunhofer distance~\cite{Selvan2017}, $d_f = 0.0625$m), these S-R and R-D distance values correspond to a path loss attenuation of around 20 dB and 80 dB from the source to the relay and from the relay to the destination respectively~\cite{seybold2005}. We assume an energy harvesting receiver sensitivity of -40dBm and a power loss due to active SI cancellation of 13mW~\cite{Otis2013}\cite{Sayed2017}. Energy harvesting receiver sensitivity denotes the minimum power level required to activate the circuit components at the receiver. Thermal noise floor is assumed to be -100 dBm, where the RSI noise/SNR loss in the baseband is assumed to be 1 dB unless otherwise stated~\cite{Katti2014}.

The following numerical simulations are averaged over 5000 independent channel realizations, where each element of the channel matrices, $\mathbf{H}$ and $\mathbf{G}$, is generated as $w_{ij} = (1/d)^{\gamma} u _{ij}$, with $u_{ij}\sim \mathcal{CN} (0,1)$ and $d$ as the distance between respective nodes. For example, for a transmit power of $P_s = 1$ W $= 0$ dBW, the Rayleigh fading channels have variances -20 dB and -80 dB for the S-R and R-D channel respectively. The loopback self-interference link is assumed to have Rician fading with shape parameter or K-factor $=30$ dB and effective channel gain, including the effect of path loss/passive physical isolation between the transmit and receive chains sharing antennas at the relay through a circulator device, as $\Omega = -20$dB~\cite{Sabharwal2012}\cite{Korpi2016}. The Rician channel $\bm F$ is modeled as the sum of a fixed component and a variable (or scattered) component as below
\begin{equation*}
\boldsymbol{F} = \sqrt{\Omega} \Big ( \sqrt{\frac{K}{K+1}} \boldsymbol{F_o} + \sqrt{\frac{1}{K+1}} \boldsymbol{F_w} \Big )
\end{equation*}
Here, $\sqrt{\frac{K}{K+1}} \boldsymbol{F_o} = \mathbb{E}[\boldsymbol{F}]$ is the line of sight component of the channel, $\sqrt{\frac{1}{K+1}} \boldsymbol{F_w}$ is the fading component that assumes uncorrelated fading. The elements of $\boldsymbol{F_o}$ have unit power, with its structure dependent on antenna configurations such as polarity, spacing, antenna type, and often has full-rank, whereas the elements for $\boldsymbol{F_w}$ are i.i.d $\mathcal{CN}(0, 1)$~\cite{Paulraj2003}\cite{Alouini2006}.

\subsection{Algorithm Convergence and Complexity}
\begin{figure}[t]
\begin{center}
\includegraphics[scale=0.5]{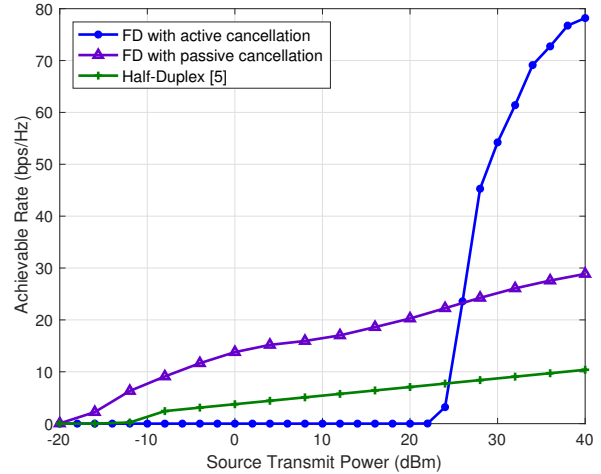}
\caption{Performance comparison with residual Self Interference directly proportional to $P_r$ for a 4x4x4 MIMO system, showing gains over half-duplex communication~\cite{Malik2018} and the different affects of active and passive self-interference cancellation.}
\label{fig9}
\end{center}
\end{figure}
Figure~\ref{converg} shows the average run-time performance versus the number of relay antennas for our proposed algorithm in comparison to several others. The results show that our proposed algorithm for full-duplex transmission with non-uniform power splitting is many times faster than the convex standard solver CVX~\cite{cvx}, where the difference comes from the efficient primal variables updates using closed form equations and the use of ellipsoid algorithm in ours in contrast to the heuristic successive approximation method of CVX. Figure~\ref{converg} also shows a comparison of our proposed algorithm with uniform power splitting to existing sequential methods for a half-duplex relay~\cite{Alouini2016}, which solve for both hops separately as a split problem and use an iterative grid search to find the uniform power splitting ratio. Note that our proposed uniform power splitting algorithm also includes the additional feature of self-interference harvesting which does not exist in~\cite{Alouini2016} and nonetheless it shows superior run-time performance.

We also include the average run time of the algorithm for half-duplex transmission with non-uniform power splitting in~\cite{Malik2018}, which follows similar steps as the proposed full-duplex algorithm in this work yet exhibits much faster convergence as seen in Figure~\ref{converg}. This difference speaks for the additional computational complexity due to the full-duplex feature, particularly in the relay power allocation factors and relay power splitting ratios as caused by the self-interference channel. This is evident in the update equations for both the power allocation in (9b) and the power splitting ratio in (10), which requires summation over products of elements in the rotated self-interference matrix. The \textit{water-level} for FD scheme is varying as seen in (9b) where as it is constant for the half-duplex transmission~\cite[Equation 8b]{Malik2018}. This added analytical and computational complexity causes the proposed full-duplex algorithm to converge slower than the half-duplex case. Such complexity is inherent to the full-duplex structure.

\subsection{Effect of Residual Self Interference in the baseband}
\begin{figure}[t]
\begin{center}
\includegraphics[scale=0.5]{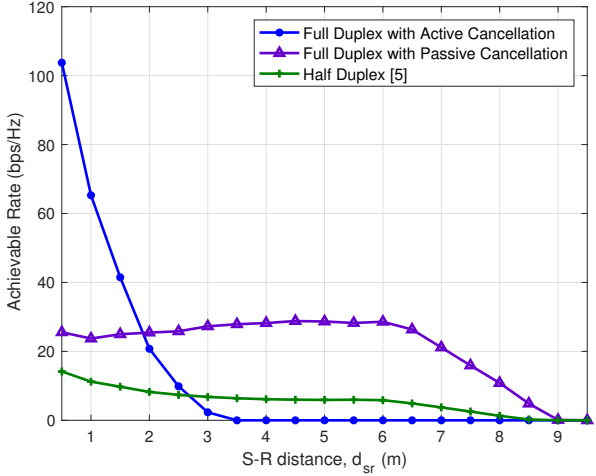}
\caption{Achievable rates versus source-relay distance in a 2x2x2 MIMO system with $P_s = 25$dBm. Active cancellation is optimal when the relay is closer to the source and passive cancellation is optimal as the relay moves away. }
\label{fig10}
\end{center}
\end{figure}
Figure~\ref{fig9} shows a comparison between half-duplex communication and two implementations of full-duplex communication - using active cancellation with $P_{IC} \neq 0$ and using passive cancellation only with $P_{IC} = 0$. For the former implementation, the residual self interference in baseband corresponds to a constant raise in the noise floor, while for the latter this residual self interference is directly proportional to the transmit power $P_r$ of the relay. For these results, with residual self interference $\sigma_f^2 = \alpha P_r^\beta$, we use $\beta = 1$ and $\alpha = 10^{-4}$. Using $\beta = 1$ implies that the residual interference increases linearly with the relay transmit power $P_r$ and the chosen value of $\alpha$ corresponds to a $40$dB interference cancellation. This can be achieved through a common passive device, like a circulator, which is able to provide 20-40 dB isolation between the transmit and receive chains~\cite{Korpi2016}\cite{Kiayani2018}.

We see that for low source transmit powers, the active cancellation FD is unable to harvest sufficient power for cancellation and transmission, and therefore achieves zero rate for $P_s \leq 20$ dBm. On the other hand, the FD scheme with passive cancellation offers enhanced rate performance over half-duplex communication for low as well as high source transmit powers. Comparing between the active and passive cancellation full-duplex schemes, we see a sharp increase in rate for $P_s > 20$ dBm for the active cancellation scheme, with significant rate gain over both half-duplex and passive full-duplex schemes. In fact, as the source power is above a threshold, the active FD scheme harvests power sufficient enough for both active cancellation as well as transmission, and since the residual interference with active cancellation is reduced to the noise floor, transmission rate is significantly enhanced. On the other hand, with increased source power, for the passive cancellation scheme, the residual self interference in baseband also increases, and therefore the rate performance is sluggish. This plot suggests an interesting observation that a hybrid self-interference cancellation methodology could be implemented, where in the low SNR regime, passive cancellation only should suffice, whereas in the medium and high SNR regime, active cancellation helps in maximally utilizing the benefits of full-duplex communication.

\subsection{Effect of distance between the source and relay}
\begin{figure}[t]
\begin{center}
\includegraphics[scale=0.5]{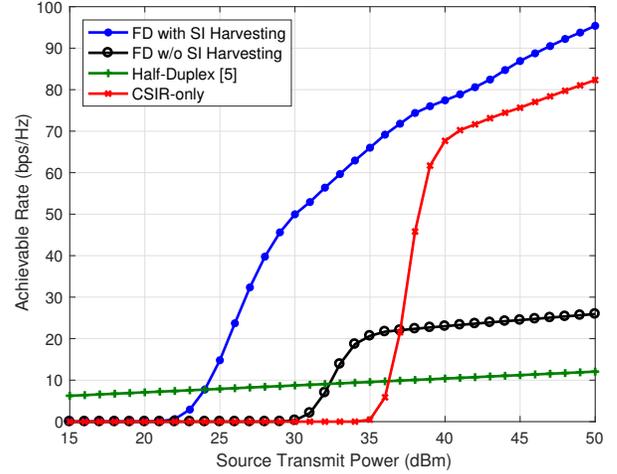}
\caption{Performance Comparison of the proposed FD scheme with SI harvesting with (i) a full-duplex scheme without SI power harvesting, (ii) a full-duplex scheme with CSI at receiver (CSIR) only, and (iii) a half-duplex scheme~\cite{Malik2018} for a 4x4x4 MIMO system}
\label{fig2}
\end{center}
\end{figure}

We consider a scenario where the source and destination location is fixed, with $d_{sd} = 10$m, and the relay is gradually moved away from the source towards the destination. The source transmit power is fixed at $P_s = 25$dBm, and a 2x2x2 MIMO system is considered. For this scenario, we compare the performance of half-duplex communication with full-duplex communication for two cases - passive self-interference cancellation only, and both passive and active cancellation. Figure~\ref{fig10} shows that even with passive self-interference cancellation only, full-duplex transmission always outperforms half-duplex communication. When the relay is closer to the source, such that sufficient power can be harvested, the superior performance of the active cancellation scheme is noteworthy. However as the distance between source and relay is increased beyond a certain length, it renders the harvested power to be insufficient for the relay to sustain both active cancellation and data transmission in the second hop. Therefore, again we could argue for a hybrid cancellation policy, where for the relay closer to the source (large $P_r$), active cancellation is preferable and for larger distances between the source and relay, passive cancellation should be employed. Using this hybrid cancellation policy would ensure that full-duplex transmission would always outperform half-duplex transmission at a maximal gain.

\subsection{Performance comparison to other SWIPT schemes}
\begin{figure}[t]
\begin{center}
\includegraphics[scale=0.5]{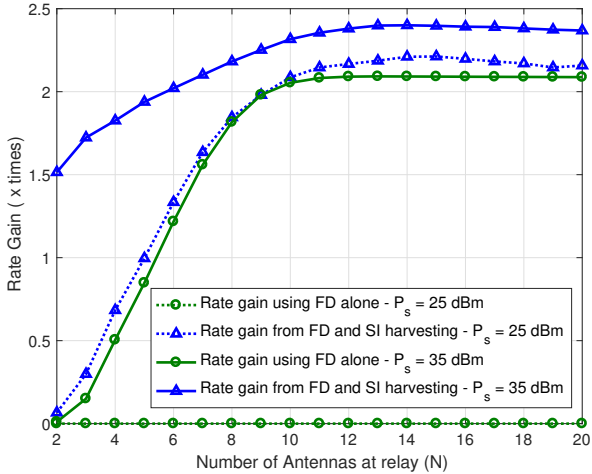}
\caption{Mutiplicative rate gain over the half-duplex scheme by (i) Full-duplex communication alone (ii) FD communication with SI harvesting for 2xNx2 MIMO systems}
\label{fig2a}
\end{center}
\end{figure} 
Figure~\ref{fig2} shows a comparison of our proposed scheme to (i) a full-duplex scheme without SI power harvesting, (ii) a full-duplex scheme with CSIR only, and (iii) a half-duplex scheme using the traditional half-duplex MIMO channel as in~\cite{Malik2018} with non-uniform power splitting. In this set of simulations and all following subsections, the baseband residual self-interference power is set to a constant regardless of the source transmit power, as discussed in Section II.B. For the FD scheme without SI harvesting, values for the channel gains from self-interference channel are set to zero to remove the component of SI power harvesting from the expression for $P_r$ when calculating the values for the primal variables as given in equations (9a), (9b) and (10). Results in Figure 3 allow us to analyze several aspects of the proposed scheme as discussed next in the following subsections.

\subsection{Benefit of full-duplex transmission and self-interfernce harvesting}
From Figure~\ref{fig2}, we observe that at both low and high source transmit powers, our proposed FD scheme with self-interference power harvesting significantly outperforms the full-duplex transmission with no self-interference harvesting, demonstrating the immense benefit of self-interference harvesting in FD communications.

Figure~\ref{fig2a} shows the multiplicative rate gains achieved at fixed source transmit powers, $P_s = 35$ dBm and $P_s = 25$ dBm, with a fixed number of antennas at source and destination $(N_s = N_d = 2)$, and an increasing number of antennas at the relay $(N)$. We separately show the multiplicative gain from using full duplex over half-duplex communication alone, and from both self-interference harvesting and full-duplex transmission to emphasize the effect of each aspect. 
\begin{figure}[t]
\begin{center}
\includegraphics[scale=0.5]{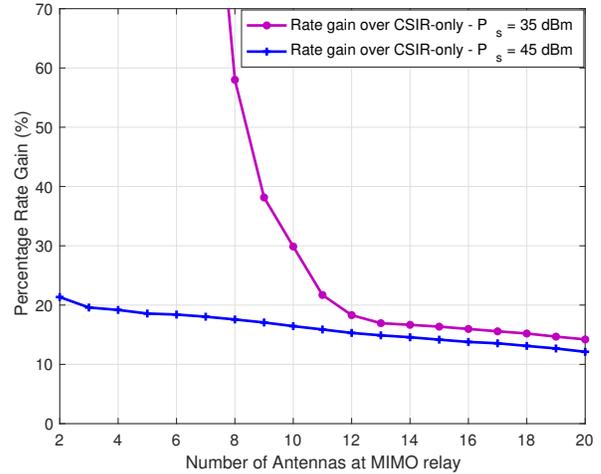}
\caption{Percentage rate gain of full CSIR over CSIR-only for 2xNx2 MIMO systems}
\label{fig2b}
\end{center}
\end{figure}
At a moderate source transmit power of $P_s = 35$ dBm we see the rate gain of our proposed scheme, with full-duplex communication and self-interference power harvesting, gradually increasing and settling to a constant at more than $2.5\times$ achievable rate of the half-duplex scheme.  About half of this gain comes from full-duplex transmission and the other half from self-interference harvesting. Consistent with Figure~\ref{fig2}, at lower source transmit power of $P_s = 25$ dBm, we see that the multiplicative rate gain for the full-duplex scheme (with SI harvesting) at $N \leq 3$ is less than 1, which implies that the half-duplex scheme outperforms the FD scheme in this region because the harvested energy in FD system here is not enough to power both signal transmission and active self-interference cancellation. Interestingly the rate gain from FD transmission alone at $P_s = 25$ dBm is zero across all simulated relay antenna settings, since the harvested power from the source signal alone is insufficient and all the actual rate gain at this low source transmit power comes from self-interference harvesting. As the number of antennas at the relay node increases, higher self-interference allows better energy harvesting and leads to an increase in the rate gain of the FD scheme with self-interference energy harvesting to more than $2\times$ the half-duplex rate. 

We thus see that the rate gain of FD over half-duplex transmission is drastic provided a sufficient number of relay antennas; further this gain increases with higher source transmit power. Thus FD transmission enables multiple folds in throughput gain via both the efficient spectral usage and self-interference harvesting. 

\subsection{Benefit of CSI at transmitter}
Comparing in Figure~\ref{fig2} between our proposed scheme, with self-interference harvesting and precoding design, and the CSIR only case, where transmitting nodes allocate equal power across all eigenmodes, we see that the difference in performance is more pronounced at lower transmit power. As the SNR increases, the difference in throughput gradually reduces to a constant performance gap. Even with CSIR only, the FD scheme with non-uniform power splitting still achieves significant throughput gain over half-duplex transmission, demonstrating the feasibility of FD information transfer and energy harvesting with just CSIR.

Figure~\ref{fig2b} demonstrates the percentage rate gain achieved by our scheme with full-duplex transmission and pre-coding design over the FD scheme with CSIR-only at fixed source transmit power levels, with two antennas at the source and destination, and an increasing number of antennas at the relay (N). This percentage rate gain is calculated as: $\%\text{RateGain}=\frac{(R_{CF}-R_{CR})}{R_{CR}} ×100\%$ , where $R_{CF}$ denotes the achievable rate of our proposed scheme in the full CSI case and $R_{CR}$ denotes the achievable rate of our scheme in the CSIR-only case. Consistent with Figure~\ref{fig2}, at $P_s = 35$ dBm for lower number of antennas, the CSIR-only case has nearly zero rate since the harvested power is not sufficient for active cancellation and signal transmission, leading to almost infinite rate gain by our scheme. As the number of antennas is increased, however, we see that the rate gain of full-CSI over CSIR-only case reduces to a constant $\sim$ 15\%. At the increased source transmit power of 45 dBm, the CSIR-only scheme can function at all relay antenna settings and the performance gap between full-CSI and CSIR-only reduces even further to just between 10-12\%. At 25dBm source power, however, the CSIR-only rate is zero for all simulated relay antenna settings and is therefore not shown. These results show that the value of having CSI at the transmitters is most pronounced at low source transmitter power and/or small number of relay antennas. As the source power or the number of relay antennas increases, the difference in throughput settles to a small gap and may justify the practical approach of having CSIR only.
  
\subsection{Benefits of using multiple antennas}
\begin{figure}[t]
\begin{center}
\includegraphics[scale=0.5]{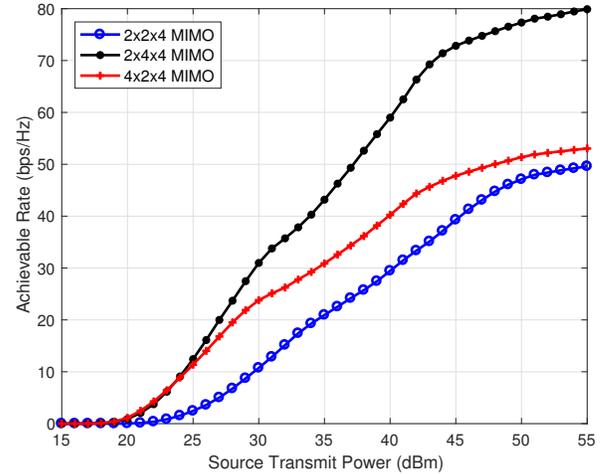}
\caption{Throughput comparison among different MIMO configurations, showing the benefit of having more antennas at the relay.}
\label{fig3}
\end{center}
\end{figure}
Figure~\ref{fig3} shows a comparison between a MIMO system with two antennas at both the source and relay, to the case when we employ four antennas at the relay, $N_r = 4$, and when we have four transmit antennas at the source node, $N_s = 4$. The number of antennas at the destination is fixed at $N_d = 4$. Both the 4x2x4 and 2x4x4 MIMO configurations offer higher gains than the 2x2x4 MIMO case because of the additional antennas at the source or relay. The 2x4x4 MIMO in particular has the advantage of diversity gain as well as power gains~\cite{Tse2005}, and furthermore, having more antennas at the relay corresponds to higher self-interference power harvesting. We therefore observe that the 2x4x4 MIMO system achieves the highest throughput, significantly higher than the other two systems. Furthermore, from Figure~\ref{fig2a}, we see that increasing the number of relay antennas up to a certain point brings out more benefit of full duplexing and self-interference harvesting, and the benefit diminishes after that point. These results suggest there exists an optimal number of relay antennas (in relation to and slightly higher than the number of source and destination antennas) that a full-duplex self-interference harvested system should use.
\begin{figure}[t]
\begin{center}
\includegraphics[scale=0.5]{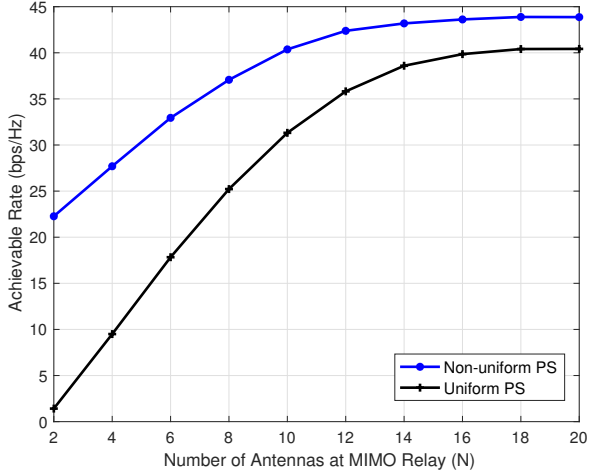}
\caption{Comparison between FD non-uniform and uniform power splitting schemes with increasing number of relay antennas at $P_s = 35$ dBm, $N_s = N_d = 2$}
\label{fig4}
\end{center}
\end{figure}

\subsection{Comparison to uniform power splitting}
Figure~\ref{fig4} shows how the proposed non-uniform power splitting scheme offers significant rate gain over uniform power splitting especially at a moderate number of antennas. Here we fix the number of antennas at source and destination at two $(N_s = N_d = 2)$, and increase the number of antennas at the relay. The rate gain of non-uniform over uniform power splitting is pronounced at a small to moderate number of relay antennas but reduces as the number of relay antennas increases. This is due to the channel hardening effect observed when the number of relay antennas is significantly larger than the number of antennas at the source and destination~\cite{Narasimhan2014}. As the channel hardens for the tall matrix $\boldsymbol{H}$, the square matrix $\boldsymbol{H^\ast H}/N_r$ converges to $\boldsymbol{I}$~\cite{Tarokh2004}, such that the channel eigenvalues approach 1; similarly for the wide matrix $\boldsymbol{G}$. Thus the channel eigenvalues for the Rayleigh fading S-R and R-D channels converge and become almost similar, $\lambda_{H,i} \approx \lambda_{H}$, $\lambda_{G,i} \approx \lambda_{G}$. Even though the self-interference channel $\boldsymbol{F}$ is Rician, $\boldsymbol{\tilde{F}_{kj} \tilde{F}_{kj}^\ast}$ in the expression for $P_r$ in (7) correspond to the modified matrix $\tilde{\boldsymbol{F}} = \boldsymbol{U_H^\ast FV_G}$, where $\boldsymbol{U_H}$ and $\boldsymbol{V_G}$ are the left and right singular vector matrices of the S-R and R-D Rayleigh channels $\boldsymbol{H}$ and $\boldsymbol{G}$ respectively, therefore the converging eigenvalues of $\boldsymbol{H}$ and $\boldsymbol{G}$ also affect the matrix $\boldsymbol{\tilde{F}}$. Thus as the number of relay antennas alone increases, both power allocation and power splitting approach to being uniform due to similar channel eigenvalues in (7). However, while the performance gap reduces, the two rates do not converge, and we see that non-uniform power splitting maintains a rate higher than that achieved by uniform power splitting.

Figure~\ref{fig5} shows a comparison of the two power splitting schemes for different MIMO configurations as the source transmit power increases. We see that using non-uniform power splitting, we can obtain non-zero rate at lower source transmit powers as opposed to uniform power splitting, and the gap widens significantly as the number of antennas increases. As the SNR increases, both splitting schemes converge; however, at standard transmit powers, for example around 25 - 35 dBm for WiFi, the non-uniform power splitting scheme achieves very low rates even in the 8x8x8 MIMO configuration, which necessitates non-uniform power splitting in this case.

\section{Conclusion}
\begin{figure}[t]
\begin{center}
\includegraphics[scale=0.5]{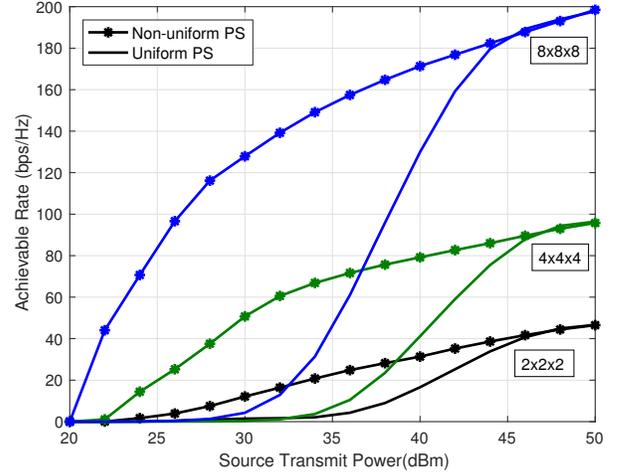}
\caption{Rate performance of uniform and non-uniform power splitting schemes with increasing source transmit power and different antenna configurations (source x relay x destination)}
\label{fig5}
\end{center}
\end{figure}
This paper investigated a full-duplex MIMO relay channel with a self-sustained relay harvesting power from both source and self-interference signals using a non-uniform power splitting technique. We formulated a rate optimization problem to jointly optimize precoders power allocation and relay power splitting, then designed an efficient primal-dual algorithm to solve it. Our analysis showed the optimality of non-uniform power splitting in a MIMO system. We showed how harvesting self-interference power in the RF domain can prove beneficial for full duplex relays, with significant rate gains over half-duplex systems. Numerical results demonstrated that by using larger MIMO systems or higher transmit power at the source, the harvested energy can potentially cater for both the power consumed by modern active cancellation circuits and the power needed for relaying transmission. These results also suggest that a hybrid self-interference cancellation policy with passive cancellation only at low SNR and added active cancellation at higher SNR can realize the maximum benefit of full-duplex communication over half-duplex transmission. We further analyzed the performance of our scheme for the case with CSI present at receiving nodes only, and showed that with more antennas at the relay or at moderately high source transmit powers, full-duplex MIMO gains with energy harvesting can be exploited even with no precoding design at the transmitting nodes.

\appendix
\subsection{Proof of Lemma~\ref{equality}}\label{appendixA}
Combining constraints (5d) and (5e) in (P) we get
\begin{align*}
&\text{tr} \left (\left(\boldsymbol{I - F^\ast Q_r^\ast \Lambda_{\rho} Q_r F} \right) \boldsymbol{W_r} \right ) \leq \text{tr} \left(\boldsymbol{\Lambda_{\rho} Q_r H W_s H^\ast Q^\ast}  \right) - P_{IC} \tag{A1}
\end{align*}
\begin{enumerate}[leftmargin=*]
\item We first consider the optimizing variable $\boldsymbol{W_s}$ which appears in constraints (5a), (5c) and the combined constraint (A1). Note that the term $\boldsymbol{Q_r H W_s H^\ast Q_r^\ast}$ appears in both (5a) and (A1). Without affecting constraint (5c), since $\boldsymbol{\Lambda_{\rho}}$ is a diagonal matrix, by applying the Hamadard inequality and an inequality relating the trace of a matrix product to the sum of eigenvalue products (\cite{marshall1979inequalities}, Chapter 9), 
both the right hand side of (5a) and (A1) are simultaneously maximized if we choose
\begin{align*}
&\boldsymbol{Q_r = U_H^\ast}\\
&\boldsymbol{W_s = V_H \Lambda_{s} V_H^\ast} \tag{A2}
\end{align*}
where $\boldsymbol{U_H}$ and $\boldsymbol{V_H}$ are obtained from the singular value decomposition of S-R channel matrix as $\boldsymbol{H = U_H \Sigma_H V_H^\ast}$. The constraint (5a) then reduces to
\begin{align*}
&R \leq \log_2 \left | \boldsymbol{I} + \frac{\boldsymbol{(I - \Lambda_{\rho}) \Sigma_H \Lambda_s \Sigma_H^\ast}}{\sigma_p^2 + \sigma_f^2}\right | \tag{A3}
\end{align*}
In short, $ \boldsymbol{\Lambda_{\rho} Q_r H W_s H^\ast Q^\ast}$ is completely diagonalized. Similarly with $\boldsymbol{\Lambda_H \triangleq \Sigma_H \Sigma_H^\ast}$, (5d) reduces to
\begin{align*}
\text{tr} \left (\left(\boldsymbol{I - F^\ast U_H \Lambda_{\rho} U_H^\ast F} \right) \boldsymbol{W_r} \right ) \leq \text{tr} \left(\boldsymbol{\Lambda_{\rho} \Lambda_s \Lambda_H}  \right) - P_{IC} \tag{A4}
\end{align*}
\item Next we consider the optimizing variable $\boldsymbol{W_r}$ for which it is more straightforward to use the original constraints in (5b), (5d) and (5e) instead of (A1). We re-write (5e) as
\begin{align*}
&P_r = \text{tr} \left(\Lambda_{\rho} \left(\boldsymbol{\Lambda_s \Lambda_H + U_H^\ast F W_r F^\ast U_H} \right) \right) - P_{IC}\tag{A5}
\end{align*}
Considering constraints (5b), (5d) and (5e), because of the non-uniform power splitting matrix $\boldsymbol{\Lambda_{\rho}}$ in (5e), the eigenvectors of $\boldsymbol{W_r}$ can be chosen to maximize the right hand side of (5b) subject to (5c), without the need to consider (5e). Essentially, we fix the relay transmit beamforming vectors to maximize the transferred/harvested power at the relay and the optimality of the solution is not affected since non-uniform power splitting provides the necessary degrees of freedom for optimization. Maximizing the right hand side of (5b) subject to (5d) is then waterfilling with
\begin{align*}
&\boldsymbol{W_r = V_G \Lambda_r V_G^\ast} \tag{A6}
\end{align*}
where $\boldsymbol{V_G}$ is obtained from the singular value decomposition of the R-D channel matrix as $\boldsymbol{G = U_G \Sigma_G V_G^\ast}$ and the optimal transmit beamforming vectors from the relay are $\boldsymbol{V_G}$.
It is worth noting here that with uniform power splitting which we discuss for comparison in Section IV-B, the transmit beam vectors from the relay will depend on both channels $\boldsymbol{G}$ and $\boldsymbol{F}$. However, in this paper, for fair comparison, we choose $\boldsymbol{W_r}$ to have the same structure as in non-uniform power splitting, such that $\boldsymbol{W_r = V_G \Lambda_r V_G^\ast}.$
\item Next we consider the relay transmit power, $P_r$ which appears in constraints (5d) and (5e). With $\boldsymbol{W_s}$ and $\boldsymbol{Q_r}$ as given in (A2), $P_r$ has the form as given in (A5). Focusing on the second term in the expression for $P_r$ in (A5), then by substituting $\boldsymbol{W_r}$ from (A6), we have
\begin{align*}
P_r &= \text{tr} \left(\Lambda_{\rho} \boldsymbol{\Lambda_s \Lambda_H} \right ) +\text{tr} \left (\boldsymbol{\Lambda_{\rho} U_H^\ast F V_G \Lambda_r V_G^\ast F^\ast U_H} \right )- P_{IC}\\
&= \text{tr} \left(\Lambda_{\rho} \boldsymbol{\Lambda_s \Lambda_H} \right ) + \text{tr} \left (\boldsymbol{\Lambda_{\rho} \tilde{F} \Lambda_r \tilde{F}^\ast} \right ) - P_{IC}
\end{align*}
where the auxiliary matrix $\boldsymbol{\tilde{F} \triangleq U_H^\ast F V_G}$. For a diagonal matrix $\boldsymbol{D}$ and a general matrix $\boldsymbol{A}$, we can use the matrix identity $\text{tr} \left( \boldsymbol{AD}\right) = \text{tr} \left( \text{diag}(\boldsymbol{A}) \boldsymbol{D}\right)$ and re-write $P_r$ as
\begin{align*}
P_r = \text{tr} \left( \Lambda_{\rho} \boldsymbol{\Lambda_s \Lambda_H} \right ) + \text{tr} \left (\boldsymbol{\Lambda_{\rho}} \text{diag} \left( \boldsymbol{\tilde{F} \Lambda_r \tilde{F}^\ast} \right) \right ) - P_{IC}
\end{align*}
Next, applying the following matrix identity for a diagonal matrix D and a general matrix $\boldsymbol{A}$ 
$$\left( \boldsymbol{A D A^\ast}\right)_{kk} = \sum_{i} d_i a_{ki} a_{ki}^\ast$$
we reach the equivalent expression in constraint (6d).
\end{enumerate}
\subsection{Proof of Lemma~\ref{convexity}}\label{appendixB}
The objective function $R$ is linear.
\begin{itemize}[leftmargin=*]
\item For constraint (6a), we define $g(\rho_i, p_i) = 1 + (1 - \rho_i) \lambda_{H,i} p_i$, which is neither convex nor concave, since its Hessian, $\nabla^2 g = [0 \ \ -\lambda_{H,i}; -\lambda_{H,i} \ \ 0]$, is indefinite; with eigenvalues $\pm \lambda_{H,i}$. The superlevel sets, $\{(\rho_i, p_i) \in \mathbb{R}_{+}^2), g(\rho_i p_i) \geq t \}$, are convex for all $t$, which makes $g(\rho_i p_i)$ a quasi-concave function~\cite{Boyd2004}. Applying the implicit constraints; $\rho_i, p_i \geq 0$, \textbf{dom} $g(\rho_i, p_i) \subset \mathbb{R}^2_{+}$, the composition function, $f = h \circ g$ in (6a), of the non-decreasing function $h(\rho_i, p_i) = \log(g(\rho_i, p_i))$ and quasi-concave function $g(\rho_i, p_i)$, is then quasiconcave with convex superlevel sets~\cite{Peter1999}.
\item (6b) is linear in $R$ and $\log(1 + q_i \lambda_{G,i})$ is concave in $q_i$.
\item The constraint (6c) is affine.
\item For the harvested power constraint, (6d), the right-hand-side of the inequality constraint has the form $\sum_{i = 1}^{K_1} f_1(\rho_i, p_i)  = \sum_{i = 1}^{K_1} \lambda_{H,i} \rho_i, p_i$. Applying the implicit constraints on $p_i \text{ and } \rho_i$, we have $\textbf{dom} (p_i, \rho_i) \subset \mathbb{R}^2_{+}$. Similar to (6a), the Hessian given as, $\nabla ^2 f_1 = [ 0 \ \  \lambda_{H,i}; \lambda_{H,i} \ \ 0 ]$, has eigenvalues $\pm \sqrt{\lambda_{H,i}^2}$, and is indefinite, so the function is neither convex nor concave. The superlevel sets, $\big \{(\rho_i, p_i) \in \mathbb{R}_{+}^2 \big \vert f_1(\rho_i, p_i) \geq t \big \}$, are convex for all $t \geq 0$ and hence $f_1(\rho_i, p_i)$ is quasi-concave or equivalently (6d) is quasi-convex. The convex sublevel sets for the quasi-convex constraint preserve the convexity of the domain of the problem. The left hand side of (6d) has the form, $\sum_{i = 1}^{K_2} f_2(\rho_i, q_i)  = \sum_{j = 1}^{K_2} q_j - \sum_{k = 1}^{K_2}\sum_{j = 1}^{K_2}  \rho_k q_j \boldsymbol{\tilde{F}_{kj} \tilde{F}_{kj^\ast}}$. The second summation term has the same form as $f_1(\rho_i, p_i)$, and is therefore quasi-convex. The first summation term $\sum_{i = 1}^{K_2} q_i$ is a linear sum and is therefore convex.
\end{itemize}
Problem (P2) is then optimization of a convex objective function over a convex set and convex sublevel sets, and is hence a convex optimization problem~\cite[p.~136-138]{Boyd2004}. 

\subsection{Proof of Theorem 1}\label{appendixC}
Setting $\nabla_{p_i} \mathcal{L} = 0$ from the Lagrangian in (8) to obtain
\begin{align*}
&\frac{\delta \mathcal{L}}{\delta p_i} = \frac{\alpha (1-\rho_i)\lambda_{H,i}}{\sigma_p^2 +\sigma_f^2 + (1-\rho_i) \lambda_{H,i} p_i} - \nu + \mu \rho_i \lambda_{H,i} = 0\\
&\iff \alpha (1-\rho_i)\lambda_{H,i} - \left( \nu + \mu \rho_i \lambda_{H,i} \right ) \\
& \times \left( \sigma_p^2 +\sigma_f^2 + (1-\rho_i) \lambda_{H,i} p_i\right)=0 \\
&\iff \left( \nu + \mu \rho_i \lambda_{H,i} \right )(1-\rho_i)\lambda_{H,i} p_i \\
&\hphantom{\iff}= \alpha (1-\rho_i)\lambda_{H,i} - \left( \nu + \mu \rho_i \lambda_{H,i} \right ) \left( \sigma_p^2 +\sigma_f^2\right)\\
&\iff p_i = \frac{1}{\left( \nu + \mu \rho_i \lambda_{H,i} \right )(1-\rho_i)\lambda_{H,i}}\big [\alpha (1-\rho_i)\lambda_{H,i} \\
&\hphantom{\iff p_i = \frac{1}{\left( \nu + \mu \rho_i \lambda_{H,i} \right )(1-\rho_i)\lambda_{H,i}}}-  \left( \nu + \mu \rho_i \lambda_{H,i} \right ) \left( \sigma_p^2 +\sigma_f^2 \right) \big ]
\end{align*}
Combining with the boundary conditions leads to $p_i^\star$ in (9a).
 
Setting $\nabla_{q_i} \mathcal{L} = 0$
\begin{align*}
&\frac{\delta \mathcal{L}}{\delta q_i} = \frac{(1- \alpha)\lambda_{G,i}}{\sigma_d^2+ \lambda_{G,i} q_i} - \mu \left (1 - \sum_{k=1}^{K_2} \rho_k \boldsymbol{\tilde{F}}_{kj} \boldsymbol{\tilde{F}}_{kj}^\ast \right ) = 0\\
&\iff (1 - \alpha) \lambda_{G,i} - \mu \left (1 - \sum_{k=1}^{K_2} \rho_k \boldsymbol{\tilde{F}}_{kj} \boldsymbol{\tilde{F}}_{kj}^\ast \right ) (\sigma_d^2 + \lambda_{G,i} q_i) = 0\\
&\iff q_i \big [\mu \lambda_{G,i} \left (1 - \sum_{k=1}^{K_2} \rho_k \boldsymbol{\tilde{F}}_{kj} \boldsymbol{\tilde{F}}_{kj}^\ast \right ) \big] \\
&\hphantom{\iff}= (1 - \alpha) \lambda_{G,i} -\mu \sigma_d^2 \left (1 - \sum_{k=1}^{K_2} \rho_k \boldsymbol{\tilde{F}}_{kj} \boldsymbol{\tilde{F}}_{kj}^\ast \right)
 \end{align*}
Re-arranging the terms and combining with boundary conditions leads to the expression for $q_i^\star$ in (9b).

Setting $\nabla_{\rho_i} \mathcal{L} = 0$ 
 \begin{align*}
 &\frac{\delta \mathcal{L}}{\delta \rho_i} = \frac{-\alpha p_i \lambda_{H,i}}{\sigma_p^2 + \sigma_f^2 +(1-\rho_i) p_i \lambda_{H,i}} + \mu \sum_{k=1}^{K_2} q_k \boldsymbol{\tilde{F}}_{ik} \boldsymbol{\tilde{F}}_{ik}^\ast + \mu \lambda_{H,i} p_i  = 0 \\
 &\iff -\alpha p_i \lambda_{H,i} + \left( \mu \sum_{k=1}^{K_2} q_k \boldsymbol{\tilde{F}}_{ik} \boldsymbol{\tilde{F}}_{ik}^\ast + \mu \lambda_{H,i} p_i \right) \\
 &\hphantom{\iff \ \ } \left( \sigma_p^2 + \sigma_f^2 + (1-\rho_i) p_i \lambda_{H,i}\right) = 0\\
 &\iff (1 - \rho_i) (\lambda_{H,i} p_i) \left( \mu \sum_{k=1}^{K_2} q_k \boldsymbol{\tilde{F}}_{ik} \boldsymbol{\tilde{F}}_{ik}^\ast + \mu \lambda_{H,i} p_i \right) \\
 & = \alpha \lambda_{H,i} p_i - (\sigma_p^2 + \sigma_r^2)\left( \mu \sum_{k=1}^{K_2} q_k \boldsymbol{\tilde{F}}_{ik} \boldsymbol{\tilde{F}}_{ik}^\ast + \mu \lambda_{H,i} p_i \right)
 \end{align*} 
Re-arranging the terms leads to equation for $\rho_i^\star$ in (10).
 

\bibliographystyle{./IEEEtran}
\bibliography{./jsacbib}
\end{document}